\documentclass[11pt]{article}

\usepackage{dsfont}

\def\colorful{0}

\usepackage{amsthm,amsfonts,amsmath,amssymb,epsfig,color,float,graphicx,verbatim, enumitem}
\usepackage{multirow}
\usepackage{algorithm}
\usepackage[noend]{algpseudocode}

\newif\ifhyper\IfFileExists{hyperref.sty}{\hypertrue}{\hyperfalse}
\hypertrue
\ifhyper\usepackage{hyperref}\fi

\usepackage{enumitem}

\usepackage{framed}
\usepackage{nicefrac}

\def\nnewcolor{1}
\ifnum\nnewcolor=1

\fi
\ifnum\nnewcolor=0

\fi

\ifnum\colorful=1
\newcommand{\new}[1]{{\color{red} #1}}

\else
\newcommand{\new}[1]{{#1}}

\fi

\newtheorem{theorem}{Theorem}[section]

\newtheorem{lemma}[theorem]{Lemma}
\newtheorem{informal theorem}[theorem]{Theorem (informal statement)}

\newtheorem{proposition}[theorem]{Proposition}
\newtheorem{corollary}[theorem]{Corollary}

\newtheorem{fact}[theorem]{Fact}

\theoremstyle{definition}
\newtheorem{definition}[theorem]{Definition}

\newcommand{\dtv}{d_{\mathrm {TV}}}
\newcommand{\Var}{\mathbf{Var}}
\newcommand{\E}{\mathbf{E}}
\newcommand{\Cov}{\mathbf{Cov}}
\newcommand{\eps}{\epsilon}
\newcommand{\tr}{\mathrm{tr}}

\newcommand{\R}{\mathbb{R}}

\newcommand{\pr}{\mathbf{Pr}}
\newcommand{\poly}{\mathrm{poly}}

\newcommand{\D}{\mathcal{D}}

\newcommand{\wh}{\widehat}

\title{Recent Advances in \\ Algorithmic High-Dimensional Robust Statistics\footnote{This article is an
expanded version of an invited chapter entitled ``Robust High-Dimensional Statistics'' in the
book ``Beyond the Worst-Case Analysis of Algorithms'' to be published by Cambridge University Press.}}

\author{
Ilias Diakonikolas\thanks{Supported by NSF Award CCF-1652862 (CAREER) and a Sloan Research Fellowship.}\\
University of Wisconsin-Madison\\
{\tt ilias@cs.wisc.edu}\\
\and
Daniel M. Kane\thanks{Supported by NSF Award CCF-1553288 (CAREER) and a Sloan Research Fellowship.}\\
University of California, San Diego\\
{\tt dakane@cs.ucsd.edu}\\
}

\begin{document}

\maketitle

\begin{abstract}
Learning in the presence of outliers is a fundamental problem in statistics.
Until recently, all known efficient unsupervised learning algorithms were very sensitive to outliers in high dimensions.
In particular, even for the task
of 
robust mean estimation under natural distributional assumptions,
no efficient algorithm was known.
Recent work in theoretical computer science gave the first efficient robust estimators for a
number of fundamental statistical tasks, including mean and covariance estimation.
Since then, there has been a flurry of research activity on algorithmic
high-dimensional robust estimation in a range of settings.
In this survey article, we introduce the core ideas and algorithmic techniques
in the emerging area of algorithmic high-dimensional robust statistics
with a focus on robust mean estimation. We also provide an overview
of the approaches that have led to computationally efficient robust estimators
for a range of broader statistical tasks and discuss new directions and opportunities for future work.
\end{abstract}

\section{Introduction} \label{sec:intro}

\subsection{Background} \label{ssec:background}

Consider the following basic statistical task: Given $n$ independent samples from an unknown mean
spherical Gaussian distribution $\mathcal{N}(\mu, I)$
on $\R^d$, estimate its mean vector {$\mu$} within small $\ell_2$-norm.
It is not hard to see that the empirical mean has $\ell_2$-error at most $O(\sqrt{d/n})$ from $\mu$ with high probability.
Moreover, this error upper bound is best possible among all $n$-sample estimators.

The Achilles heel of the empirical estimator is that it crucially relies on the assumption that the
observations were generated by a spherical Gaussian. The existence of even a {\em single}
outlier can arbitrarily compromise this estimator's performance.
However, the Gaussian assumption is only ever approximately valid, as real datasets are typically exposed to some
source of contamination. Hence, any estimator that is to be used in practice
must be {\em robust} in the presence of outliers.


Learning in the presence of outliers is an important goal in
statistics and has been studied in the robust statistics community
since the 1960s~\cite{Tukey60, Huber64} (see~\cite{HampelEtalBook86, Huber09} for introductory
statistical textbooks on the topic). In recent years, the problem of designing robust and computationally
efficient estimators for various high-dimensional statistical tasks has become
a pressing challenge in a number of data analysis applications.
These include the analysis of biological datasets, where natural outliers are
common~\cite{RP-Gen02, Pas-MG10, Li-Science08}
and can contaminate the downstream statistical analysis,
and {\em data poisoning attacks} in machine learning~\cite{Barreno2010}, where even a small
fraction of fake data (outliers) can substantially degrade
the quality of the learned model~\cite{BiggioNL12, SteinhardtKL17}.

Classical work in robust statistics
pinned down the 
minimax risk of high-dimensional robust estimation in several basic settings of interest.
{In contrast, until very recently, even the most basic computational questions in this field
were poorly understood.}
For example, the Tukey median~\cite{Tukey75} is a sample-efficient robust mean estimator
for spherical Gaussian distributions. However, it is NP-hard
to compute in general~\cite{JP:78} and the 
heuristics proposed to approximate it degrade
in the quality of their approximation as the dimension scales.
Similar hardness results have been shown~\cite{Bernholt, HardtM13} for
essentially all known classical estimators in robust statistics.

Until recently, all known computationally efficient estimators
could only tolerate a negligible fraction of outliers in high dimensions,
even for the basic task of mean estimation. Recent work
by Diakonikolas, Kamath, Kane, Li, Moitra, and Stewart~\cite{DKKLMS16}, and by
Lai, Rao, and Vempala~\cite{LaiRV16}
gave the first efficient robust estimators for various high-dimensional unsupervised learning
tasks, including mean and covariance estimation. Specifically,~\cite{DKKLMS16}
obtained the first polynomial-time robust estimators with {\em dimension-independent} error guarantees,
i.e., with error scaling only with the fraction of corrupted samples and {\em not} with the dimensionality of the data.
Since the dissemination of these works~\cite{DKKLMS16, LaiRV16}, there has been significant research activity
on designing computationally efficient robust estimators in a variety of settings.

\subsection{Contamination Model} \label{ssec:contam}

Throughout this article, we focus on the following model of data contamination
that generalizes several other existing models:
\begin{definition}[Strong Contamination Model] \label{def:adv}
Given a parameter $0< \eps < 1/2$ and a distribution family $\mathcal{D}$ on $\R^d$,
the \emph{adversary} operates as follows: The algorithm specifies a
number of samples $n$, and $n$ samples are drawn from some unknown $\new{X} \in \mathcal{D}$.
The adversary is allowed to inspect the samples, remove up to $\eps n$ of them
and replace them with arbitrary points. This modified set of $n$ points is then given as input
to the algorithm. We say that a set of samples is {\em $\eps$-corrupted}
if it is generated by the above process.
\end{definition}

The contamination model of Definition~\ref{def:adv} can be viewed as
a semi-random input model:
First, nature draws a set $S$ of i.i.d. samples from a statistical model of interest,
and then an adversary is allowed to change the set $S$ in a bounded way to obtain
an $\eps$-corrupted set $T$. The parameter $\eps$ is the proportion of contamination
and quantifies the power of the adversary. Intuitively, among our samples,
an unknown $(1-\eps)$ fraction are generated from a
distribution of interest and are called {\em inliers}, and the rest are called {\em outliers}.

One can consider less powerful adversaries,
giving rise to weaker contamination models. An adversary
may be (i) adaptive or oblivious to the inliers, (ii) only allowed to
add outliers, 
or only allowed to remove inliers, or allowed to do both.
We provide some examples of natural and well-studied such models in the following paragraphs.

In {\em Huber's contamination model}~\cite{Huber64}, the adversary is oblivious to the inliers
and is only allowed to add outliers. More specifically, in Huber's model, the adversary
generates samples from a mixture distribution $P$ of the form $P = (1-\eps) \new{X} + \eps N$, where
$\new{X} \in \mathcal{D}$ is the unknown target distribution and $N$ is an adversarially chosen noise distribution.
Another natural contamination model very common in theoretical computer science allows the adversary to perturb the distribution
$X$ by at most $\eps$ in total variation distance, i.e., the adversary generates samples
from a distribution $Y$ that satisfies $\dtv(Y, X) \leq \eps$.  Intuitively,
the adversary in this model is oblivious to the inliers and is allowed to both add outliers
and remove inliers. This model is very similar to a contamination model
proposed by Hampel~\cite{Hampel71}. We note that contamination in total variation distance
is strictly stronger than Huber's model. More broadly, one can naturally modify this model to study
model misspecification with respect to different loss functions (see, e.g.,~\cite{ZHS19}).

In computational learning theory, the contamination model of Definition~\ref{def:adv}
is related to the {\em agnostic model}~\cite{Haussler:92, KSS:94},
where the goal is to learn a labeling function whose agreement
with some underlying target function is close to the best possible,
among all functions in some given class.
An important difference with our setting is that the agnostic model requires that we fit all the data,
while in our robust setting we only want to fit the uncorrupted data.

The strong contamination model can be viewed as the unsupervised analogue of
the challenging {\em nasty noise model}~\cite{BEK:02} (itself a strengthening of the malicious
model~\cite{Valiant:85short, keali93}).
In the nasty model, an adversary is allowed to corrupt an $\eps$-fraction of both the labels
and the samples, and the goal of the learning algorithm
is to output a hypothesis with small misclassification error
with respect to the clean data generating distribution.




In robust mean estimation, given an $\eps$-corrupted set
of samples from {a well-behaved distribution} (e.g., $\mathcal{N}(\mu, I)$),
we want to output a vector $\widehat{\mu}$ that closely approximates the unknown mean vector $\mu$.
A natural choice of metric between the means for identity covariance distributions is the
$\ell_2$-error $\|\widehat{\mu}-\mu \|_2$, and in this article we focus on designing robust estimators
minimizing $\|\widehat{\mu}-\mu \|_2$. We note however that the same algorithms typically lead to small
Mahalanobis distance, i.e., $\|\widehat{\mu}-\mu \|_{\Sigma} = |(\widehat{\mu}-\mu)^T \Sigma^{-1} (\widehat{\mu}-\mu)|^{1/2}$,
when the underlying distribution has covariance $\Sigma$.

One can analogously define robust estimation for other parameters
of high-dimensional distributions (e.g., covariance matrix and higher-order moment tensors)
with respect to natural loss functions (e.g., Frobenius norm, spectral norm).
A more general statistical task is that of {\em robust density estimation}:
Given an $\eps$-corrupted set of samples from an unknown distribution $\new{X} \in \mathcal{D}$,
output a hypothesis distribution $H$ (not necessarily in $\mathcal{D}$) such that the total variation distance $\dtv(H, \new{X})$
is minimized. We note that robust density estimation and robust parameter estimation are closely related to each other. For many
natural parametric distributions, the latter can be reduced to the former for an appropriate choice of metric
between the parameters.

In all these settings, the goal is to design computationally efficient learning algorithms
that achieve {\em dimension-independent} error, i.e.,
error that scales only with the contamination parameter $\eps$, but not with the dimension $d$.
The information-theoretic limits of robust estimation depend on our assumptions about the distribution
family $\mathcal{D}$. In the following subsection, we provide the basic relevant background.

\subsection{Basic Background: Sample Efficient Robust Estimation} \label{ssec:sample-robust}

Before we proceed with presenting computationally efficient robust estimators in the next sections,
we provide a few basic facts on the information-theoretic limits of robust estimation.
For concreteness, we focus here on robust mean estimation. We note that similar arguments
can be applied to various other parameter estimation tasks. The interested reader is referred to~\cite{DKKLMS16}
and to~\cite{chen2018, liu2019} (and references therein) for recent
information-theoretic work from the statistics community.


We first note that some assumptions on the underlying distribution family $\mathcal{D}$
are necessary for robust mean estimation to be information-theoretically possible.
Consider for example, the family $\mathcal{D} = \{D_x, x \in \R\}$, where $D_x$
is a probability distribution on the real line with only one point $x \in \R$ having positive mass
$\pr[D_x = x] = \eps >0$ and such that $\E[D_x] = x$. While estimating the mean of an arbitrary
distribution in $\mathcal{D}$ is straightforward without corruptions (by taking samples until we
see a sample twice, which must be the true mean), an adversary can erase all
information about the mean in an $\eps$-corrupted sample from $D_x$. Indeed, 
an adversary can delete the samples at $x$ and
move them at an arbitrary location to arbitrarily corrupt the sample mean.

Typical assumptions on the family $\mathcal{D}$ are either parametric (e.g., $\mathcal{D}$ could be the family
of all Gaussian distributions) or are defined by concentration properties (e.g., each distribution in $\mathcal{D}$ satisfies sub-gaussian concentration) or conditions on low-degree moments (e.g.,  each distribution in $\mathcal{D}$ has appropriately bounded
higher-order moments).

Another basic observation is that, in contrast to the uncorrupted setting,
in the contaminated setting of Definition~\ref{def:adv}
it is {\em not} possible to obtain consistent estimators ---  that is,
estimators with error converging to zero in probability as the sample size increases indefinitely.
Typically, there is an information-theoretic
limit on the minimum attainable error that depends on the proportion of contamination $\eps$
and structural properties of the underlying distribution family.

In particular, for the robust mean estimation of a high-dimensional Gaussian, we have:

\begin{fact} \label{fact:error-limit}
For any $d \geq 1$,  any robust estimator for the mean of $ \new{X}= \mathcal{N}(\mu_X, I)$ on $\R^d$,
must have $\ell_2$-error $\Omega(\eps)$, even in Huber's contamination model.
\end{fact}

This fact can be shown as follows:
Given two distinct distributions $\mathcal{N}(\mu_1, I)$ and $\mathcal{N}(\mu_2, I)$ with $\|\mu_1-\mu_2\|_2  \new{= \Theta(\eps)}$,
the adversary constructs two noise distributions $N_1, N_2$ on $\R^d$ such that
$Y = (1-\eps)\mathcal{N}(\mu_1, I) + \eps N_1 = (1-\eps)\mathcal{N}(\mu_2, I) + \eps N_2$.
Consequently, even in the infinite sample regime,
any algorithm can at best learn that its samples are coming from $Y$,
but will be unable to distinguish between the cases
where the real distribution is $\mathcal{N}(\mu_1, I)$
and where it is $\mathcal{N}(\mu_2, I)$.
This proves Fact~\ref{fact:error-limit}.


If the target distribution $\new{X}$ is allowed to come from a broader class of distributions (such as allowing any distribution with subgaussian tails, or any distribution with bounded covariance), the situation is even worse. If one can find two distributions $\new{X}$ and $\new{X}'$ in the desired class with $\dtv(\new{X},\new{X}')\leq \eps$, it becomes information-theoretically impossible for an algorithm to distinguish between the two. However, if the difference between $\new{X}$ and $\new{X}'$ is concentrated in the tails of the distribution, then $\new{X}$ and $\new{X}'$ might have very different means.
This implies that for the class of distributions with sub-gaussian tails (and identity covariance)
we cannot hope to learn the mean to $\ell_2$-error better than $\Omega(\eps\sqrt{\log(1/\eps)})$;
and for the class of distributions with covariance $\Sigma$ bounded by $\sigma^2 I$, we
cannot expect to do better than $\Omega(\sigma \sqrt{\eps})$. It turns out that these bounds
are information-theoretically optimal, and in fact as we will see, the means of such distributions
can be robustly estimated to these errors in many cases.


\medskip

The problem of robust mean estimation for $\mathcal{N}(\mu, I)$ seems so innocuous that one could naturally
wonder why simple approaches do not work. In the one-dimensional case,
it is well-known that the median is a robust estimator of the mean, matching the lower bound of Fact~\ref{fact:error-limit}.
Specifically, it is easy to show that the median $\widehat{\mu}$ of a multiset of $n = \Omega(\log(1/\tau)/\eps^2)$
$\eps$-corrupted samples from a one-dimensional Gaussian $\mathcal{N}(\mu, 1)$
satisfies $|\widehat{\mu} - \mu| \leq O(\eps)$ with probability at least $1-\tau$.

In high dimensions, the situation is more subtle.
There are many reasonable ways to attempt to generalize the median as
a robust estimator in high dimensions, but unfortunately, most natural ones lead
to $\ell_2$-error of $\Omega(\eps \sqrt{d})$ in $d$ dimensions, even in the infinite sample regime
(see, e.g.,~\cite{DKKLMS16, LaiRV16}).

Perhaps the most obvious high-dimensional generalization of the median is the {\em coordinate-wise median}.
Here the $i$-th coordinate of the output is the median of the $i$-th
coordinates of the input samples. This estimator guarantees that every coordinate of
the output is within $O(\eps)$ of the corresponding coordinate of the
input. This estimator suffices for obtaining small $\ell_{\infty}$-error, but if one wants
$\ell_2$-error (which is natural in the case of Gaussians), then there exist instances
where the coordinate-wise median has $\ell_2$-error as large as $\Omega(\eps \sqrt{d})$.
Another potential way to generalize the median to high dimensions is via the {\em geometric median},
i.e., the point $x^{\ast}$ that minimizes $\sum_{i} \|x^{(i)} - x^{\ast}\|_2$. Unfortunately,
the geometric median can also produce $\ell_2$-error of $\Omega(\eps \sqrt{d})$ if the adversary
adds the $\eps$-fraction of the outliers all off from the mean in the same
direction.

A third high-dimensional generalization of the median relies on the observation that
taking the median of {\em any} univariate projection of our input points
gives us an approximation to the projected mean. Finding a mean vector
that minimizes the error over the {\em worst direction} can actually
be used to obtain $\ell_2$-error of $O(\eps)$ with high probability.
In other words, it is possible to reduce the high-dimensional robust mean estimation
problem to a collection of one-dimensional robust mean estimation problems.
This is the underlying idea in Tukey's median~\cite{Tukey75}, which is known to be
a robust mean estimator for spherical Gaussians and for more general symmetric distributions. But unfortunately, the Tukey median relies on combining information for infinitely many directions, and is unsurprisingly NP-Hard to compute in general.

The following proposition gives a computationally inefficient robust mean estimator
matching the lower bound of Fact~\ref{fact:error-limit}:

\begin{proposition} \label{prop:sample-robust}
There exists an algorithm that, on input an $\epsilon$-corrupted set of samples
from $\new{X} = \mathcal{N}(\mu_X, I)$ of size $n = \Omega((d+\log(1/\tau))/\eps^2)$,
runs in $\poly(n, 2^d)$ time, and
outputs $\widehat{\mu} \in \R^d$ such that with probability at least $1-\tau$, it holds that
$\|\widehat{\mu}  - \mu_X \|_2 = O(\eps)$.
\end{proposition}

The algorithm to establish Proposition~\ref{prop:sample-robust} proceeds by
using a one-dimensional robust mean estimator to estimate $v \cdot \mu$,
for an appropriate net of $2^{O(d)}$ unit vectors $v \in \R^d$, and
then combines these estimates (by solving a large linear program)
to obtain an accurate estimate of $\mu$.
When $\new{X} = \mathcal{N}(\mu_X, I)$, our one-dimensional robust mean estimator
will be the median, giving the $O(\eps)$ error in Proposition~\ref{prop:sample-robust}.

We note that the same procedure is applicable to other distribution families as well (even non-symmetric ones),
as long as there is an accurate robust mean estimator for each univariate projection.
Specifically, if $X$ has tails bounded by those of a Gaussian with covariance $ \sigma^2 I$,
one can use the {\em trimmed mean} for each univariate projection. This gives error of $O(\sigma \eps \sqrt{\log(1/\eps)} )$.
If $X$ is only assumed to have bounded covariance matrix ($\Sigma_X \preceq \sigma^2 I$),
we can similarly use the trimmed mean, which leads to total error of $O(\sigma \sqrt{\eps})$.
By the discussion following Fact~\ref{fact:error-limit}, both these upper bounds are optimal, within constant factors,
under the corresponding assumptions.


\subsection{Structure of this Article} \label{ssec:structure}

In Section~\ref{sec:robust-mean}, we provide a unified presentation of
two related algorithmic techniques that gave
the first computationally efficient algorithms for high-dimensional robust mean estimation.
Section~\ref{sec:robust-mean} is the main technical section of this article
and showcases a number of core ideas and techniques that can be applied
to several high-dimensional robust estimation tasks. Section~\ref{sec:gen}
provides an overview of recent algorithmic progress for more general robust estimation tasks.
Finally, in Section~\ref{sec:conc} we conclude with a few general directions for future work.

\subsection{Preliminaries and Notation} \label{ssec:prelims}

For a distribution $X$, we will use the notation $x \sim X$
to denote that $x$ is a sample drawn from $X$. For a finite set $S$, we will write $x \sim_u S$ to
denote that $x$ is drawn from the uniform distribution on $S$.
The probability of event $\mathcal{E}$ will be denoted by $\pr[\mathcal{E}]$.

We will use $\E[X]$ and $\Var[X]$ to denote the expectation and the variance of random variable $X$.
If $X$ is multivariate, we will denote by $\Cov[X]$ its covariance matrix.
We will also use the notation $\mu_X$ and $\Sigma_X$ for the mean and covariance of $X$.
Similarly, for a finite set $S$, we will denote by $\mu_S$ and $\Sigma_S$ the sample
mean and sample covariance of $S$.

For a vector $v$, we will use $\|v\|_2$ to denote its $\ell_2$-norm. For a matrix $A$, we will
denote by $\|A\|_2$ and $\|A\|_F$ its spectral and Frobenius norms respectively, and
by $\tr(A)$ its trace.
We will denote by $\preceq$ the Loewner order between matrices.

We will use standard asymptotic notation $O(\cdot)$, $\Omega(\cdot)$.
The $\tilde{O}(\cdot)$ notation hides logarithmic factors in its argument.

\section{High-Dimensional Robust Mean Estimation} \label{sec:robust-mean}
In this section, we illustrate the main insights underlying
recent robust high-dimensional learning algorithms
by focusing on the problem of robust mean estimation.
The algorithmic techniques presented in this section were introduced in~\cite{DKKLMS16, DKK+17}.
Here we give a simplified and unified presentation that illustrates the key ideas
and the connections between them.

The objective of this section is to provide the intuition and
background required to develop robust learning algorithms
in an accessible way. As such, we will not attempt to optimize
the sample or computational complexities of {the} algorithms {presented},
other than to show that they are polynomial in the relevant parameters.

In the problem of robust mean estimation, we are given an $\eps$-corrupted set of samples
from a distribution $X$ on $\R^d$ and our goal is to approximate the mean of $X$, within small
error in $\ell_2$-norm. In order for such a goal to be information-theoretically possible,
it is required that $X$ belongs to a suitably well-behaved family of distributions.
A typical assumption is that $X$ belongs to a family whose moments are guaranteed to satisfy certain conditions,
or equivalently, a family with appropriate concentration properties.
In our initial discussion, we will use the running example of a spherical Gaussian,
although the results presented here hold in greater generality.
That is, the reader is encouraged to imagine that $X$ is of the form $\mathcal{N}(\mu,I)$,
for some unknown $\mu\in \R^d$.

\paragraph{Structure of this Section.}
In Section~\ref{ssec:intuition}, we discuss the basic intuition underlying {the presented} approach.
In Section~\ref{ssec:goodsets}, we will describe a stability condition
that is necessary for the algorithms in this section to
succeed.
In the subsequent subsections, we present two related algorithmic techniques taking advantage of the stability
condition in different ways. Specifically, in Section~\ref{ssec:convex-program}, we describe an algorithm
that relies on convex programming. In Section~\ref{ssec:filter},
we describe an iterative outlier removal technique, which has been the method of choice in practice.
In Section~\ref{ssec:lit-robust-mean}, we conclude with an overview of the relevant literature.

\subsection{Key Difficulties and High-Level Intuition}\label{ssec:intuition}

Arguably the most natural idea to robustly estimate the mean of a distribution
would be to identify the outliers and output the empirical mean of the remaining points.
The key conceptual difficulty is the fact
that, in high dimensions, the outliers cannot be identified at an individual level {even when
they move the mean significantly}.
In many cases, we can easily identify the ``extreme outliers'' ---
via a pruning procedure exploiting the concentration properties of the inliers.
Alas, such naive 
approaches typically do not suffice to obtain non-trivial error guarantees.

The simplest example illustrating this difficulty is that of a high-dimensional spherical Gaussian.
Typical samples will be at $\ell_2$-distance approximately $\Theta(\sqrt{d})$ from the true mean.
That is, we can certainly identify as outliers all points of our dataset at distance more than $\Omega(\sqrt{d})$
from the coordinate-wise median of the dataset. All other points cannot be removed via such a procedure,
as this could result in removing many inliers as well. However, by placing an $\eps$-fraction of outliers
at distance $\sqrt{d}$ in the same direction from the unknown mean, an adversary can corrupt the
sample mean by as much as $\Omega(\eps \sqrt{d})$.

This leaves the algorithm designer with a dilemma of sorts.
On the one hand, potential outliers at distance $\Theta(\sqrt{d})$ from the unknown mean
could lead to large $\ell_2$-error, scaling polynomially with $d$.
On the other hand, if the adversary places outliers at distance
approximately $\Theta(\sqrt{d})$ from the true mean in {\em random directions},
it may be information-theoretically impossible to distinguish them from the inliers.
The way out is the realization that {\em it is in fact not necessary to detect and remove all outliers}.
It is only required that the algorithm can detect the ``consequential outliers", i.e., the ones that
can significantly impact our estimates of the mean. 


Let us assume without loss of generality that there no extreme outliers (as these can be removed via
pre-processing). Then {\em the only way that the empirical mean can be far from the true mean is
if there is a ``conspiracy'' of many outliers, all producing errors in approximately the same direction.}
Intuitively, if our corrupted points are at distance $O(\sqrt{d})$ from the true mean
in random directions, their contributions will on average cancel out,
leading to a small error in the sample mean. In conclusion, it suffices 
to be able to detect these kinds of conspiracies of outliers.

The next key insight is simple and powerful.
Let $T$ be an $\epsilon$-corrupted set of points drawn from $\mathcal{N}(\mu, I)$.
If such a conspiracy of outliers substantially moves the empirical mean $\widehat{\mu}$ of $T$,
it must move $\widehat{\mu}$ in some direction. That is, there is a unit vector $v$ such
that these outliers cause $v\cdot(\widehat{\mu}-\mu)$ to be large.
For this to happen, it must be the case that these outliers are on average far from $\mu$ in the $v$-direction.
In particular, if an $\eps$-fraction of corrupted points in $T$ move the sample average of
$v\cdot (\new{U_T}-\mu)$, where $\new{U_T}$ is the uniform distribution on $T$, by more than $\delta$
($\delta$ should be thought of as small, but substantially larger than $\eps$),
then on average these corrupted points $x$ must have $v\cdot(x-\mu)$ at least $\delta/\eps$.
This in turn means that these corrupted points will have a contribution of at least
$\eps \cdot (\delta/\eps)^2 = \delta^2/\eps$ to the variance of $v \cdot \new{U_T}$. Fortunately,
this condition can actually be algorithmically detected! In particular, by computing
the top eigenvector of the sample covariance matrix, we can efficiently determine whether
or not there is any direction $v$ for which the variance of $v\cdot \new{U_T}$ is abnormally large.

The aforementioned discussion leads us to the overall structure of the algorithms {we will describe in this section}.
Starting with an $\eps$-corrupted set of points $T$ (perhaps weighted in some way), we
compute the sample covariance matrix and find the eigenvector $v^{\ast}$ with largest eigenvalue $\lambda^{\ast}$.
If $\lambda^{\ast}$ is not much larger than what it should be (in the absence of outliers),
by the above discussion, the empirical mean is close to the true mean, and we can return that as an answer.
Otherwise, we have obtained a particular direction {$v^{\ast}$} for which
we know that the outliers play an unusual role, i.e.,
behave significantly differently than the inliers. The distribution of the points projected
in the {$v^{\ast}$}-direction can then be used to perform some sort of outlier removal.
{The outlier removal procedure can be quite subtle and crucially depends
on our distributional assumptions about the clean data.}

\subsection{Good Sets and Stability}\label{ssec:goodsets}

In this section, we give a deterministic condition on the {uncorrupted} data
that we call \emph{stability} (Definition~\ref{def:stability}),
which is necessary for the algorithms presented here to succeed.
Furthermore, we provide an efficiently checkable condition
under which the empirical mean is certifiably close to the true mean (Lemma~\ref{lem:stability}).

Let $S$ be a set of $n$ i.i.d. samples drawn from $X$.
We will typically call these sample points good.
The adversary can select up to an $\eps$-fraction of points in $S$ and replace them with arbitrary points
to obtain an $\eps$-corrupted set $T$, which is given as input to the algorithm.
To establish correctness of an algorithm, 
we need to show that with high probability over the choice
of the set $S$, for any choice of how to corrupt the good samples that
the adversary makes, the algorithm will output an accurate estimate of the target mean.

To carry out such an analysis, it is convenient to explicitly state {a
collection of sufficient deterministic conditions on the set $S$. Specifically},
we will define a notion of a ``stable'' set,
quantified by the proportion of contamination $\eps$ and
the distribution $X$. The precise stability conditions vary considerably
based on the underlying estimation task and the assumptions on the distribution family of the uncorrupted data.
Roughly speaking, we require
that the uniform distribution over a stable set $S$ behaves similarly to the distribution $X$ with respect to
higher moments and, potentially, tail bounds. Importantly, we require that these conditions hold
even after removing an arbitrary $\eps$-fraction of points in $S$.

The notion of a stable set must have two critical properties:
(1) A set of $n$ i.i.d. samples from $X$ is stable with high probability,
when $n$ is at least a sufficiently large polynomial in the relevant parameters; and
(2) If $S$ is a stable set and $T$ is obtained from $S$ by changing at most an $\eps$-fraction of the points in $S$,
then {the} algorithm when run on the set $T$ will succeed.

The robust mean estimation algorithms that will be presented in this section
crucially rely on considering sample means and covariances. The following
stability condition is an important ingredient in the success criteria of {these algorithms:}
\begin{definition}[Stability Condition] \label{def:stability}
Fix $0< \eps<1/2$ and $\delta \geq \eps$.
A finite set $S \subset \R^d$ is \emph{$(\epsilon,\delta)$-stable} (with respect to a distribution $X$)
if for every unit vector $v \in \R^d$ and every $S'\subseteq S$ with $|S'| \geq (1-\epsilon)|S|$,
the following conditions hold:
\begin{enumerate}
\item $\left|\frac{1}{|S'|}\sum_{x\in S'} v\cdot( x-\mu_X)\right| \leq \delta \;,$ and
\item $\left|\frac{1}{|S'|}\sum_{x\in S'} (v\cdot( x-\mu_X))^2 - 1\right| \leq \delta^2/\epsilon.$
\end{enumerate}
\end{definition}


The aforementioned stability condition or a variant thereof is used
in every known robust mean estimation algorithm.
Definition~\ref{def:stability} requires that after restricting to a $(1-\eps)$-density subset $S'$,
the sample mean of $S'$ is within $\delta$ of \new{the mean of $X$, $\mu_X$,}
and the sample variance of $S'$ is $1\pm \delta^2/\epsilon$ in every direction.
\new{(We note that Definition~\ref{def:stability} is intended for distributions $X$ with covariance
$\Sigma_X = I$ or, more generally, $\Sigma_X \preceq I$. The case of arbitrary known or bounded covariance
can be reduced to this case via an appropriate linear transformation of the data.)}

The fact that the conditions of Definition~\ref{def:stability}
must hold {\em for every} large subset $S'$ of $S$ might make
it unclear if they can hold with high probability. However, one can show the following:

\begin{proposition} \label{prop:gaussian-good-set}
A set of i.i.d. samples from an identity covariance sub-gaussian distribution of size $\Omega(d/\eps^2)$ is
$(\eps, O(\eps\sqrt{\log(1/\eps)})$-stable with high probability.
\end{proposition}

We sketch a proof of Proposition~\ref{prop:gaussian-good-set}.
The only property required for the proof is that the distribution of the uncorrupted data
has identity covariance and sub-gaussian tails in each direction, i.e., the tail probability of each
univariate projection is bounded from above by the Gaussian tail.

Fix a direction $v$.
To show that the first condition holds for this $v$,
we note that we can maximize $\frac{1}{|S'|}\sum_{x\in S'} v\cdot( x-\mu_X)$
by removing from $S$ the $\eps$-fraction of points $x$ for which $v\cdot x$ is smallest.
Since the empirical mean of $S$ is close to $\mu_X$ with high probability,
we need to understand how much this quantity is altered by removing the $\eps$-tail in the $v$-direction.
{Given our assumptions on the distribution of the uncorrupted data},
removing the $\eps$-tail only changes the mean by $O(\eps\sqrt{\log(1/\eps)})$.
Therefore, if the empirical distribution of $v\cdot x$, $x \in S$, behaves like a spherical Gaussian
in this way, the first condition is satisfied.

The second condition follows via  a similar analysis. We can minimize the relevant
quantity by removing the $\eps$-fraction of points $x \in S$ with $|v\cdot( x-\mu_X)|$ as large as possible.
If $v\cdot x$ {is} distributed like a {unit-variance} sub-gaussian, the total mass of its square over the $\eps$-tails is
$O(\eps\log(1/\eps))$. We have thus established that both conditions hold with high probability for any fixed direction.
Showing that the conditions hold with high probability for all directions simultaneously
can be shown by an appropriate covering argument.

More generally, one can show quantitatively different stability conditions under various distributional assumptions.
In particular, we state the following proposition without proof.
(The interested reader is referred to~\cite{DKK+17} for a proof.)

\begin{proposition} \label{prop:bounded-cov-good-set}
A set of i.i.d. samples from a distribution with covariance $\Sigma \preceq I$ of size $\tilde{\Omega}(d/\eps)$ is
$(\eps, O(\sqrt{\eps}))$-stable with high probability.
\end{proposition}

We note that analogous bounds can be easily shown for {\em identity covariance} distributions with bounded higher central moments.
For example, if our distribution has identity covariance and its $k$-th central moment is bounded from above by a constant,
one can show that a set of $\Omega(d/\eps^{2-2/k})$ samples is $(\eps, O(\eps^{1-1/k}))$-stable with high probability.


The aforementioned notion of stability is powerful and suffices for robust mean estimation.
Some of the algorithms that will be presented in this section only work
under the stability condition, while others require additional conditions beyond stability.

The main reason why stability suffices is quantified in the following lemma:
\begin{lemma}[Certificate for Empirical Mean]\label{lem:stability}
Let $S$ be an $(\eps,\delta)$-stable set with respect to a distribution $X$, for some $\delta \geq \epsilon>0$.
Let $T$ be an {$\eps$-corrupted version of $S$.}
Let $\mu_T$ and $\Sigma_T$ be the empirical mean and covariance of $T$.
If the largest eigenvalue of $\Sigma_T$ is at most $1+\lambda$, for some $\lambda \geq 0$, then
$\|\mu_T-\mu_X\|_2 \leq O(\delta +\sqrt{\eps \lambda}).$
\end{lemma}

\begin{proof}[Proof of Lemma~\ref{lem:stability}.]
Let $S' = S \cap T$ and 
$T' = T \setminus S'$. By replacing $S'$ with a subset if necessary, we may assume that $|S'|= (1-\eps) |S|$ and
$|T'|=\eps |S|$. Let $\mu_{S'},\mu_{T'},\Sigma_{S'},\Sigma_{T'}$ represent the empirical
means and covariance matrices of $S'$ and $T'$. A simple calculation gives that
$$
\Sigma_T = (1-\eps) \Sigma_{S'} + \eps \Sigma_{T'} + \eps(1-\eps)(\mu_{S'}-\mu_{T'})(\mu_{S'}-\mu_{T'})^T \;.
$$
Let $v$ be the unit vector in the direction of $\mu_{S'}-\mu_{T'}$. We have that
\begin{align*}
1+\lambda \geq v^T \Sigma_T v & = (1-\eps) v^T\Sigma_{S'}v + \eps v^T\Sigma_{T'}v
+ \eps(1-\eps)v^T(\mu_{S'}-\mu_{T'})(\mu_{S'}-\mu_{T'})^Tv \\
& \geq (1-\eps) (1-\delta^2/\eps) + \eps(1-\eps) \|\mu_{S'}-\mu_{T'}\|_2^2\\
& \geq 1 - O(\delta^2/\eps) + (\eps/2) \|\mu_{S'}-\mu_{T'}\|_2^2 \;,
\end{align*}
where we used {the variational characterization of eigenvalues, the fact that $\Sigma_{T'}$ is positive
semidefinite, and} the second stability condition for $S'$.
By rearranging, we obtain that $\|\mu_{S'}-\mu_{T'}\|_2 = O(\delta/\eps+\sqrt{\lambda/\eps}).$
Therefore, we can write
\begin{align*}
\|\mu_T-\mu_X\|_2 & = \|(1-\eps)\mu_{S'}+\eps \mu_{T'} - \mu_X\|_2
= \|\mu_{S'}-\mu_X + \epsilon(\mu_{T'}-\mu_{S'})\|_2\\
& \leq \|\mu_{S'}-\mu_X\|_2 + \eps \|\mu_{S'}-\mu_{T'}\|_2
= O(\delta) + \eps \cdot O(\delta/\eps + \sqrt{\lambda/\eps})\\
& = O(\delta + \sqrt{\lambda \eps})\;,
\end{align*}
where we used the first stability condition for $S'$ and our bound on $\|\mu_{S'}-\mu_{T'}\|_2$.
\end{proof}


We note that the proof of Lemma~\ref{lem:stability} only used the lower bound part in the second
condition of Definition~\ref{def:stability}, i.e., that the sample variance of $S'$ in each direction
is at least $1-\delta^2/\eps$. The upper bound part will be crucially used in the design and analysis of
our robust mean estimation algorithms in the following sections.

Lemma~\ref{lem:stability} says that if our input set of points $T$ is {an $\eps$-corrupted version of
any stable set $S$ and has bounded sample covariance, the sample mean
of $T$ closely approximates the true mean. This lemma, or a variant thereof,
is a key result in all known robust mean estimation algorithms.

Unfortunately, we are not always guaranteed that the set $T$ we are given has this property.
{In order to deal with this, we will want to find a subset of $T$ with bounded covariance and large intersection with $S$.
However, for some of the algorithms presented, it will be convenient
to find a probability distribution over $T$ rather than a subset.
For this, we will need a slight generalization of Lemma~\ref{lem:stability}.}
\begin{lemma}\label{lem:stability-full}
Let $S$ be an $(\eps,\delta)$-stable set with respect to a distribution $X$,
for some $\delta \geq \eps>0$ {with $|S|>1/\eps$}.
Let $W$ be a probability distribution that differs from $U_S$, the uniform distribution over $S$,
by at most $\eps$ in total variation distance. Let $\mu_W$ and $\Sigma_W$ be the mean and covariance of $W$.
If the largest eigenvalue of $\Sigma_W$ is at most $1+\lambda$, for some $\lambda \geq 0$, then
$\|\mu_W-\mu_X\|_2 \leq O(\delta +\sqrt{\eps\lambda}).$
\end{lemma}

{Note that this subsumes Lemma \ref{lem:stability} by letting $W$ be the uniform distribution over $T$.
The proof is essentially identical to that of Lemma \ref{lem:stability}, except that we need to
show that the mean and variance of the conditional distribution $W | S$ are approximately correct,
whereas in Lemma~\ref{lem:stability} the bounds on the mean and variance of $S\cap T$
followed directly from stability.}

Lemma~\ref{lem:stability-full} clarifies the goal of our outlier removal procedure.
In particular, given our initial $\eps$-corrupted set $T$, we will attempt to find a distribution
$W$ supported on $T$ so that $\Sigma_W$ has no large eigenvalues. The weight
$W(x)$, $x \in T$, quantifies our belief whether point $x$ is an inlier or an outlier.
We will also need to ensure that any such $W$ we choose is close to the uniform distribution over $S$.

More concretely, we now describe a framework that captures our robust mean estimation algorithms.
We start with the following definition: 
\begin{definition} \label{def:c}
Let $S$ be a $(3\eps,\delta)$-stable set with respect to $X$ and let $T$ be
{an $\eps$-corrupted version of $S$.} Let $\mathcal{C}$ be the set of
all probability distributions $W$ supported on $T$, where $W(x) \leq \frac{1}{|T|(1-\eps)}$, for all $x \in T$.
\end{definition}
We note that \emph{any} distribution in $\mathcal{C}$ differs from $U_S$, the uniform distribution on $S$, by at most $3\eps$.
Indeed, for $\eps \leq 1/3$, we have that:
\begin{align*}
\dtv(U_S,W) & = \sum_{x\in T} \max\{ W(x)-U_S(x),0 \} \\
&= \sum_{x\in S\cap T} \max\{ W(x)-1/|T|, 0\} + \sum_{x\in T\setminus S} W(x)\\
& \leq \sum_{x\in S\cap T} \frac{\eps}{|T|(1-\eps)} + \sum_{x\in T \setminus S}\frac{1}{|T|(1-\eps)}\\
&\leq |T|\left(\frac{\eps}{|T|(1-\eps)}\right) + \eps |T| \left( \frac{1}{|T|(1-\eps)}\right)\\
& = \frac{2\eps}{1-\eps} \leq 3\eps \;.
\end{align*}
Therefore, if we find $W\in \mathcal{C}$ with $\Sigma_W$ having no large eigenvalues,
Lemma \ref{lem:stability-full} implies that
$\mu_W$ {is} a good approximation to $\mu_X$.
Fortunately, we know that such a $W$ exists! In particular, if we take $W$ to be $W^{\ast}$,
the uniform distribution over {$S\cap T$}, 
the largest eigenvalue is at most $1+\delta^2/\epsilon$, and thus we achieve $\ell_2$-error $O(\delta)$.

At this point, we have an {\em inefficient} algorithm for approximating $\mu_X$:
Find \emph{any} $W \in \mathcal{C}$ with bounded covariance.
The remaining question is how we can efficiently find one.
There are two basic algorithmic techniques to achieve this,
that we present in the subsequent subsections.

The first algorithmic technique we will describe is based on convex programming.
{We will call this {\em the unknown convex programming method}.}
Note that $\mathcal{C}$ is a convex set and that finding a point in $\mathcal{C}$
that has bounded covariance is \emph{almost} a convex program.
It is not quite a convex program, because the variance of $v\cdot W$, for fixed $v$,
is not a convex function of $W$. However, one can show that
given a $W$ with variance  in some direction significantly larger than $1+\delta^2/\eps$,
we can efficiently construct a hyperplane separating $W$ from $W^{\ast}$
{(recall that $W^\ast$ is the uniform distribution over $S\cap T$)}
(Section~\ref{ssec:convex-program}). This method has the advantage of naturally working
under only the stability assumption. On the other hand, as it relies on the ellipsoid algorithm,
it is quite slow (although polynomial time).

Our second technique, which we will call {\em filtering}, is an iterative
outlier removal method that is typically faster, as it relies primarily on spectral techniques.
The main idea of the method is the following: If $\Sigma_W$ does not have large eigenvalues,
then the empirical mean is close to the true mean. Otherwise,
there is some unit vector $v$ such that $\Var[v\cdot W]$ is substantially larger than it should be.
This can only be the case if $W$ assigns substantial mass to elements of $T \setminus S$
that have values of $v\cdot x$ very far from the true mean {of $v\cdot \mu$}. This observation allows us to perform
some kind of outlier removal, in particular by removing (or down-weighting) the points $x$ that have
$v\cdot x$ inappropriately large. An important conceptual property is that one cannot afford to remove
only outliers, but it is possible to ensure that more outliers are removed than inliers.
Given a $W$ where $\Sigma_W$ has a large eigenvalue, one filtering step gives
a new distribution $W' \in \mathcal{C}$ that is closer to $W^{\ast}$ than $W$ was.
Repeating the process eventually gives a $W$ with no large eigenvalues.
The filtering method and its variations are discussed in Section \ref{ssec:filter}.

\subsection{The Unknown Convex Programming Method}\label{ssec:convex-program}

By Lemma \ref{lem:stability-full}, it suffices to find a distribution $W\in \mathcal{C}$ with $\Sigma_W$ having no large eigenvalues.
We note that this condition \emph{almost} defines a convex program.
This is because $\mathcal{C}$ is a convex set of probability distributions and
the bounded covariance condition says that $\Var[v\cdot W] \leq 1+\lambda$ for all unit vectors $v$.
Unfortunately, the variance $\Var[v\cdot W] = \E[|v\cdot( W-\mu_W)|^2]$ is not quite linear in $W$.
(If we instead had $\E[|v\cdot( W-\nu)|^2]$, where $\nu$ is some fixed vector, this would be linear in $W$.)
However, we will show that a unit vector $v$ for which $\Var[v\cdot W]$ is too large,
can be used to obtain a separation oracle, i.e., a linear function $L$ for which $L(W)>L(W^\ast)$.

Suppose that we identify a unit vector $v$ such that $\Var[v\cdot W] = 1+\lambda$, where $\lambda > c(\delta^2/\epsilon)$
for a sufficiently large universal constant $c>0$. Applying Lemma \ref{lem:stability-full} to the one-dimensional projection
$v\cdot W$, gives $|v\cdot (\mu_W - \mu_X)| \leq O(\delta+\sqrt{\epsilon \lambda}) = O(\sqrt{\epsilon\lambda}).$

{Let $L(Y):=\E_X[|v\cdot(Y-\mu_W)|^2]$ and note that $L$ is a linear function of the probability distribution $Y$}
with $L(W)=1+\lambda$. We can write
\begin{align*}
L(W^{\ast}) & = \E_{W^{\ast}}[|v\cdot(W^{\ast}-\mu_W)|^2]
 = \Var[v\cdot W^{\ast}] + |v\cdot(\mu_W - \mu_{W^{\ast}})|^2 \\
& \leq 1+\delta^2/\eps + 2|v\cdot(\mu_W - \mu_X)|^2 + 2|v\cdot(\mu_{W^{\ast}} - \mu_X)|^2\\
& \leq 1+O(\delta^2/\eps + \eps\lambda) < 1+\lambda = L(W) \;.
\end{align*}
In summary, we have an explicit convex set $\mathcal{C}$ of probability distributions from
which we want to find one with eigenvalues bounded by $1+O(\delta^2/\eps)$.
Given any $W\in \mathcal{C}$ which does not satisfy this condition, we can produce a linear function
$L$ that separates $W$ from $W^{\ast}$.
Using the ellipsoid algorithm,
we obtain the following general theorem:

\begin{theorem} \label{thm:general-mean}
Let $S$ be a $(3\eps,\delta)$-stable set with respect to a distribution $X$
and let $T$ be {an $\eps$-corrupted version of $S$.}
There exists a polynomial time algorithm which given $T$ returns
$\widehat{\mu}$ such that $\|\widehat{\mu} - \mu_X\|_2 = O(\delta)$.
\end{theorem}

\paragraph{Implications for Concrete Distribution Families.}

Combining Theorem~\ref{thm:general-mean} with corresponding stability bounds,
we obtain a number of concrete applications for various
distribution families of interest.
Using Proposition~\ref{prop:gaussian-good-set}, we obtain:

\begin{corollary}[Identity Covariance Sub-gaussian Distributions] \label{cor:mean-subgaussian-id}
Let $T$ be an $\eps$-corrupted set of samples of size $\Omega(d/\eps^2)$
from an identity covariance sub-gaussian distribution $X$ on $\R^d$.
There exists a polynomial time algorithm which given $T$ returns
$\widehat{\mu}$ such that with high probability $\|\widehat{\mu} - \mu_X\|_2 = O(\eps \sqrt{\log(1/\eps)})$.
\end{corollary}

We note that Corollary~\ref{cor:mean-subgaussian-id} can be immediately
adapted for identity covariance distributions satisfying weaker concentration assumptions.
For example, if $X$ satisfies sub-exponential concentration in each direction, we obtain an efficient robust
mean estimation algorithm with $\ell_2$-error of $O(\eps \log(1/\eps))$. If $X$ has identity covariance
and  bounded $k$-th central moments, $k \geq 2$, we obtain error $O(\eps^{1-1/k})$. These error bounds
are information-theoretically optimal, within constant factors.

For distributions with unknown bounded covariance, using Proposition~\ref{prop:bounded-cov-good-set}
we obtain:

\begin{corollary}[Unknown Bounded Covariance Distributions] \label{cor:mean-bounded-cov}
Let $T$ be an $\eps$-corrupted set of samples of size $\tilde{\Omega}(d/\eps)$
from a distribution $X$ on $\R^d$ with unknown bounded covariance $\Sigma_X \preceq \sigma^2 I$.
There exists a polynomial time algorithm which given $T$ returns
$\widehat{\mu}$ such that with high probability $\|\widehat{\mu} - \mu_X\|_2 = O(\sigma \sqrt{\eps})$.
\end{corollary}

By the discussion following Fact~\ref{fact:error-limit}, the above error bound is also information-theoretically optimal,
within constant factors.


\subsection{The Filtering Method}\label{ssec:filter}

As in the unknown convex programming method, the goal of the filtering method is
to find a distribution $W\in \mathcal{C}$ so that $\Sigma_W$ has bounded eigenvalues.
Given a $W\in \mathcal{C}$, $\Sigma_W$ either has bounded eigenvalues (in which case the weighted empirical mean works)
or there is a direction $v$ in which $\Var[v\cdot W]$ is too large. In the latter case, the projections $v\cdot W$
must behave very differently from the projections $v\cdot S$ or $v\cdot X$. In particular,
since an $\eps$-fraction of outliers are causing a much larger increase in the standard deviation,
this means that the distribution of $v\cdot W$ will have many ``extreme points'' --- more than one would expect to find in $v\cdot S$.
This fact allows us to identify a non-empty subset of extreme points the majority of which are outliers.
These points can then be removed (or down-weighted) in order to ``clean up'' our sample.
Formally, given a $W\in \mathcal{C}$ without bounded eigenvalues, we can efficiently
find a $W' \in \mathcal{C}$ so that $W'$ is closer to $W^\ast$ than $W$ was.
Iterating this procedure eventually terminates giving a $W$ with bounded eigenvalues.

We note that while it may be conceptually useful to consider the above scheme for general distributions $W$ over points,
in most cases it suffices to consider only $W$ given as the uniform distribution over some set of points.
The filtering step in this case consists of replacing the set $T$ by some subset $T' = T \setminus R$, where $R\subset T$.
To guarantee progress towards $W^{\ast}$ (the uniform distribution over $S \cap T$),
it suffices to ensure that at most a third of the elements of $R$ are also in $S$, or equivalently that at least
two thirds of the removed points are outliers (perhaps in expectation). {The algorithm will terminate
when the current set of points $T'$ has bounded empirical covariance, and the output will be the empirical
mean of $T'$.}

Before we proceed with a more detailed technical discussion, we note that there are several possible ways to implement the filtering step, and that the method used has a significant impact on the analysis. In general, a filtering step removes
all points that are ``far'' from the sample mean in a large variance direction. However, the precise way that this is
quantified can vary in important ways.

\subsubsection{Basic Filtering} \label{ssec:filter-basic}

In this subsection, we present a filtering method that yields efficient robust mean estimators with 
optimal error bounds for identity covariance
(or, more generally, known covariance) distributions whose univariate projections
satisfy appropriate tail bounds.
For the purpose of this section, we will restrict ourselves to the Gaussian setting.
We note however that this method immediately extends to distributions with weaker concentration properties,
e.g., sub-exponential or even inverse polynomial concentration, with appropriate modifications.

We note that the filtering method presented here requires an additional condition
on our good set of samples, on top of the stability condition. This is quantified
in the following definition:

\begin{definition} \label{def:tails}
A set $S \subset \R^d$ is {\em tail-bound-good (with respect to $X=\mathcal{N}(\mu_X,I)$)}
if for any unit vector $v$, and any $t>0$, we have
\begin{equation}\label{tailEqn}
\pr_{x\sim_u S}\left[ |v\cdot(x -\mu_X)| > 2t+2 \right] \leq e^{-t^2/2} \;.
\end{equation}
\end{definition}

Since any univariate projection of $X$ is distributed like a standard Gaussian, Equation~(\ref{tailEqn})
should hold if the uniform distribution over $S$ were replaced by $X$. It can be shown that
this condition holds with high probability if $S$ consists of i.i.d. random samples from $X$
of a sufficiently large polynomial size~\cite{DKKLMS16}.


Intuitively, the additional tail condition of Definition~\ref{def:tails} is needed
to guarantee that the filtering method used here will remove more outliers than inliers.
Formally, we have the following:
\begin{lemma} \label{lem:basic-filtering}
Let $\eps>0$ be a sufficiently small constant.
Let $S \subset \R^d$ be both $(2\eps,\delta)$-stable and tail-bound-good with respect to
$X=\mathcal{N}(\mu_X,I)$, with $\delta = c \eps\sqrt{\log(1/\eps)}$, for $c>0$ a sufficiently large constant.
Let $T \subset \R^d$ be such that $|T\cap S| \geq (1-\eps)\max( |T|,|S|)$
and assume we are given a unit vector $v \in \R^d$ for which $\Var[v\cdot T] > 1+2\delta^2/\eps$.
There exists a polynomial time algorithm that returns a subset $R\subset T$ satisfying $|R\cap S| < |R|/3$.
\end{lemma}

To see why Lemma~\ref{lem:basic-filtering} suffices for our purposes,
note that by replacing $T$ by $T' = T\setminus R$, we obtain a less noisy version of $S$ than $T$ was.
In particular, it is easy to see that the size of the symmetric difference between $S$ and $T'$
is strictly smaller than the size of the symmetric difference between $S$ and $T$.
From this it follows that the hypothesis $|T\cap S| \geq (1-\eps)\max( |T|,|S|)$
still holds when $T$ is replaced by $T'$, allowing us to iterate this process
until we are left with a set with small variance.

\begin{proof}
Let $\Var[v\cdot T]=1+\lambda$.
By applying {Lemma \ref{lem:stability} to the set $T$},
we get that $|v\cdot \mu_X - v\cdot \mu_T| \leq  c \sqrt{\lambda \eps}$.
By \eqref{tailEqn}, it follows that
$\pr_{x \sim_u S}\left[ |v\cdot(x -\mu_T)| > 2t+2+c\sqrt{\lambda \epsilon} \right] \leq e^{-t^2/2}.$
We claim that there exists a threshold $t_0$ such that
\begin{equation}\label{tailviolationEqn}
\pr_{x\sim_u T}\left[ |v\cdot(x -\mu_T)| > 2t_0+2+c\sqrt{\lambda \epsilon} \right] >  4e^{-t_0^2/2} \;,
\end{equation}
{where the constants have not been optimized.}
Given this claim, the set $R=\{x\in T: |v\cdot(x -\mu_T)| > 2t_0+2+c\sqrt{\lambda \eps}\}$ will satisfy the
conditions of the lemma.

To prove our claim, we analyze the variance of $v\cdot T$ and note that
much of the excess must be due to points in $T\setminus S$. In particular, by our assumption
on the variance in the $v$-direction,
$\sum_{x\in T} |v\cdot(x -\mu_T)|^2 = |T|\Var[v\cdot T] = |T|(1+\lambda)$,
where $\lambda>2\delta^2/\eps$.
The contribution from the points $x\in S\cap T$ is at most
\begin{align*}
\sum_{x\in S} |v\cdot(x -\mu_T)|^2 & = |S| \left( \Var[v\cdot S] + |v\cdot(\mu_T-\mu_S)|^2 \right)
\leq |S|(1+\delta^2/\epsilon +2c^2\lambda\eps)\\
& \leq |T|(1+2c^2\lambda\eps+{3\lambda/5}) \;,
\end{align*}
{where the first inequality uses the stability of $S$, and the last uses that $|T|\geq (1-\eps)|S|$.}
If $\eps$ is sufficiently small relative to $c$, it follows that
$\sum_{x\in T \setminus S} |v\cdot(x -\mu_T)|^2 \geq |T|\lambda/3.$
On the other hand, by definition we have:
\begin{equation} \label{eqn:variance-outliers}
\sum_{x\in T \setminus S} |v\cdot(x -\mu_T)|^2 = |T|\int_0^\infty 2t \pr_{x\sim_u T}\left[ |v\cdot (x-\mu_T)|>t, x\not\in S \right] dt.
\end{equation}
Assume for the sake of contradiction that
there is no $t_0$ for which Equation \eqref{tailviolationEqn} is satisfied.
Then the RHS of \eqref{eqn:variance-outliers} is at most
\begin{align*}
& |T| \left(\int_0^{2+{c}\sqrt{\lambda\eps}+10\sqrt{\log(1/\eps)}} 2t \pr_{x\sim_u T}[x\not\in S] dt +
\int_{2+c\sqrt{\lambda\eps}+10\sqrt{\log(1/\eps)}}^\infty 2t \pr_{x\sim_u T}\left[ |v\cdot (x-\mu_T)|>t \right] dt \right)\\
\leq & |T|\left(\eps(2+c\sqrt{\lambda\eps}+10\sqrt{\log(1/\eps)})^2+
\int_{5\sqrt{\log(1/\eps)}}^\infty 16(2t+2+c\sqrt{\lambda\eps})e^{-t^2/2}dt  \right)\\
\leq & |T|\left(O(c^2 \lambda \eps^2 + \eps\log(1/\eps))+O(\eps^2(\sqrt{\log(1/\eps)}+c\sqrt{\lambda\eps})) \right)\\
\leq & |T|O(c^2\lambda \eps^2 + (\delta^2/\eps)/c)<  |T|\lambda/3 \;,
\end{align*}
which is a contradiction.
Therefore, the tail bounds and the concentration violation together imply
the existence of such a $t_0$ (which can be efficiently computed).
\end{proof}

\subsubsection{Randomized Filtering} \label{ssec:filter-rand}
The basic filtering method of the previous subsection is deterministic,
relying on the violation of a concentration inequality satisfied by the inliers.
In some settings, {deterministic filtering seems to fail to give optimal results
and we require the filtering procedure to be randomized.
A concrete such setting is when the uncorrupted distribution
is only assumed to have bounded covariance.

The main idea of randomized filtering is simple:
Suppose we can identify a non-negative function $f(x)$, defined on the samples $x$,
for which (under some high probability condition on the inliers) it holds that
$\sum_{T} f(x) \geq 2 \sum_{S}f(x)$, where $T$ is an $\eps$-corrupted set
of samples and $S$ is the corresponding set of inliers. Then
we can create a randomized filter by removing each sample point $x \in T$ with probability proportional to $f(x)$.
This ensures that the {\em expected} number of outliers removed is at least the {\em expected} number of inliers
removed. The analysis of such a randomized filter is slightly more subtle, so we will discuss it in the following
paragraphs.

The key property the above randomized filter ensures is that the sequence of random variables
$(\textrm{\# Inliers removed})-(\textrm{\# Outliers removed})$ {(where ``inliers'' are points in $S$ and ``outliers'' points
in $T\setminus S$)} across iterations is a sub-martingale.
Since the total number of outliers removed across all iterations accounts for at most an $\eps$-fraction of the total samples,
this means that with probability at least $2/3$, at no point does the algorithm remove more than a $2\eps$-fraction
of the inliers. A formal statement follows:

\begin{theorem} \label{thm:rand-filter}
Let $S \subset \R^d$ be a $(6\eps,\delta)$-stable set (with respect to $X$)
and $T$ be an $\eps$-corrupted version of $S$.
Suppose that given any $T' \subseteq T$ with $|T'\cap S| \geq (1-6\epsilon)|S|$
for which $\Cov[T']$ has an eigenvalue bigger than $1+\lambda$, for some $\lambda \geq 0$,
there is an efficient algorithm that computes a non-zero function
$f:T'\rightarrow \R_+$ such that $\sum_{x\in T'}f(x) \geq 2 \sum_{x\in T'\cap S} f(x)$.
Then there exists a polynomial time randomized algorithm that computes a vector $\widehat{\mu}$
that with probability at least $2/3$ satisfies $\|\widehat{\mu} - \mu_X\|_2 = O(\delta+\sqrt{\eps \lambda}).$
\end{theorem}

The algorithm is described in pseudocode below:

\bigskip

\fbox{\parbox{6in}{
{\bf Algorithm} {\tt Randomized Filtering}
\begin{enumerate}
\item Compute $\Cov[T]$ and its largest eigenvalue $\nu$.
\item If $\nu \leq 1+ \lambda$, return $\mu_T$.
\item Else
\begin{itemize}
\item Compute $f$ as guaranteed in the theorem statement.
\item Remove each $x \in T$ with probability $f(x)/\max_{x\in T} f(x)$ and return to Step $1$ with the new set $T$.
\end{itemize}
\end{enumerate}
}}

\medskip

\begin{proof}[Proof of Theorem~\ref{thm:rand-filter}]
First, it is easy to see that this algorithm runs in polynomial time.
Indeed, as the point $x \in T$ attaining the maximum value of $f(x)$ is definitely removed in each filtering iteration,
each iteration reduces $|T|$ by at least {one}.
To establish correctness, we will show that, with probability at least $2/3$,
at each iteration of the algorithm it holds $|S\cap T| \geq (1-6\epsilon)|S|$.
Assuming this claim, {Lemma \ref{lem:stability}} implies that our final error will be as desired.

To prove the desired claim, we consider the sequence of random variables
$d(T) = |S \Delta T| = |S\backslash T| + |T\backslash S|$ across the iterations of the algorithm.
We note that, initially, $d(T)=2\eps|S|$ and that $d(T)$ cannot drop below $0$.
Finally, we note that at each stage of the algorithm
$d(T)$ \new{decreases} by $(\textrm{\# Inliers removed})-(\textrm{\# Outliers removed})$,
and that the expectation of this quantity is
$$
\sum_{x\in S \setminus T}f(x) - \sum_{x\in T \setminus S}f(x) = 2\sum_{x\in S\cap T}f(x) - \sum_{x\in T}f(x) \leq 0.
$$
This means that $d(T)$ is a sub-martingale (at least until we reach a point where $|S\cap T| \leq (1-6\eps)|S|$).
However, if we set a stopping time at the first occasion where this condition fails, we note that the expectation
of $d(T)$ is at most $2\eps |S|$. Since it is at least $0$, this means that with probability at least $2/3$
it is never more than $6\eps|S|$, which would imply that $|S\cap T|\geq (1-6\epsilon)|S|$ {throughout the algorithm}.
If this is the case, the inequality $|T'\cap S| \geq (1-6\epsilon)|S|$ will continue
to hold throughout our algorithm, thus eventually yielding such
a set with the variance of $T'$ bounded. By Lemma \ref{lem:stability},
the mean of this $T'$ will be a suitable estimate for the true mean.
\end{proof}

\paragraph{Methods of Point Removal.}
The randomized filtering method described above 
{only requires} that each point $x$ is removed with probability $f(x)/\max_{x\in T} f(x)$,
without any assumption of independence.
{Therefore, given an $f$}, there are several ways to implement this scheme.
{A few natural ones are given here:}
\begin{itemize}
\item {\em Randomized Thresholding:}
Perhaps the easiest method for implementing our randomized filter
is generating a uniform random number $y\in [0,\max_{x\in T} f(x)]$
and removing all points $x \in T$ for which $f(x)\geq y$. This method is
practically useful in many applications. Finding the set of such points
is often fairly easy, as this condition may well correspond to a simple threshold.

\item {\em Independent Removal:} Each $x \in T$ is removed independently with probability $f(x)/\max_{x\in T} f(x)$.
This scheme has the advantage of leading to less variance in $d(T)$. A careful analysis of the random walk
involved allows one to reduce the failure probability to $\exp(-\Omega(\eps |S|))$.

\item {\em Deterministic Reweighting:} Instead of removing points, this scheme allows for weighted sets of points.
In particular, each point will be assigned a weight in $[0,1]$ and we will consider weighted means and covariances.
Instead of removing a point with probability proportional to $f(x)$, we can remove a fraction of $x$'s weight proportional to $f(x)$.
This ensures that the appropriate weighted version of $d(T)$ is definitely non-increasing, implying deterministic correctness of the algorithm.
\end{itemize}

\paragraph{\bf Practical Considerations.}
While the aforementioned point removal methods
have similar theoretical guarantees,
recent implementations~\cite{DiakonikolasKKLSS2018sever}
suggest that they have different practical performance on real datasets. The
deterministic reweighting method is somewhat slower in practice as its worst-case
runtime and its typical runtime are comparable. In more detail, one can guarantee
termination by setting the constant of proportionality so that at each step
at least one of the non-zero weights is set to zero. However, in practical circumstances,
we will not be able to do better. That is, the algorithm may well be forced to
undergo $\eps |S|$ iterations. On the other hand,
the randomized  versions of the algorithm are likely to remove several points of $T$
at each filtering step.

Another reason why the randomized versions may be preferable has to do with the quality of the results.
The randomized algorithms only produce bad results when there is a chance that $d(T)$ ends up being very large.
However, since $d(T)$ is a sub-martingale, this will only ever be the case if there is a corresponding possibility
that $d(T)$ will be exceptionally small. Thus, although the randomized algorithms may have a probability of giving
worse results some of the time, this will only happen if a corresponding fraction of the time,
they also give \emph{better} results than the theory guarantees. This consideration suggests
that the randomized thresholding procedure might have advantages
over the independent removal procedure precisely because it has a higher probability of failure. This has been observed
experimentally in~\cite{DiakonikolasKKLSS2018sever}: {In real datasets (poisoned with a constant fraction of
adversarial outliers), the number of iterations of randomized filtering is typically bounded by a small constant.}

\subsubsection{Universal Filtering} \label{ssec:filter-univ}

In this subsection, we show how to use randomized filtering
to construct a universal filter that works under only the stability condition {(Definition~\ref{def:stability}) ---
not requiring the tail-bound condition of the basic filter (Lemma~\ref{lem:basic-filtering})}.
Formally, we show:

\begin{proposition} \label{prop:univ-filter}
Let $S \subset \R^d$ be an $(\eps,\delta)$-stable set for $\eps, \delta>0$ sufficiently small constants
and $\delta$ at least a sufficiently large multiple of $\eps$. Let $T$ be {an $\eps$-corrupted version of}
$S$. Suppose that $\Cov[T]$ has largest eigenvalue
$1+\lambda > 1+8\delta^2/\eps$. Then there exists a computationally efficient
algorithm that, on input $\epsilon,\delta, T$, 
computes a non-zero function $f: T \rightarrow \R_+$ satisfying $\sum_{x\in T}f(x) \geq 2 \sum_{x\in T\cap S} f(x)$.
\end{proposition}

By combining Theorem~\ref{thm:rand-filter} and Proposition~\ref{prop:univ-filter}, we obtain
a filter-based algorithm establishing Theorem~\ref{thm:general-mean}.

\begin{proof}[Proof of Proposition~\ref{prop:univ-filter}.]
The algorithm to construct $f$ is the following:
We start by computing the sample mean $\mu_T$ and the top (unit) eigenvector $v$ of $\Cov[T]$.
For $x\in T$, we let $g(x)=(v\cdot (x-\mu_T))^2$. Let $L$ be the set of $\eps \cdot |T|$ elements of $T$
on which $g(x)$ is largest. We define $f$ to be $f(x) = 0$ for $x\not\in L$, and $f(x) = g(x)$ for $x\in L$.

Our basic plan of attack is as follows:
First, we note that the sum of $g(x)$ over $x\in T$ is the  variance of $v \cdot Z$, $Z \sim_u T$, which
is substantially larger than the sum of $g(x)$ over $S$, which is approximately the variance of
$v \cdot Z$, $Z \sim_u S$.
Therefore, the sum of $g(x)$ over the $\eps|S|$ elements of $T \setminus S$ must be quite large.
In fact, using the stability condition, we can show that the latter quantity
must be larger than the sum of the largest $\eps |S|$ values of $g(x)$ over $x\in S$.
However, since $|T \setminus S| \leq |L|$, we have that
$\sum_{x\in T} f(x) = \sum_{x\in L} g(x) \geq \sum_{x\in T\setminus S} g(x) \geq 2 \sum_{x\in S} f(x) \;.$

We now proceed with the detailed analysis.
First, note that $$\sum_{x \in T} g(x) = |T| \Var[v\cdot T] = |T|(1+\lambda) \;.$$
Moreover, for any $S' \subseteq S$ with $|S'| \geq (1-2\eps)|S|$, we have that
\begin{equation}\label{subsetGSumEqn}
\sum_{x\in S'} g(x) = |S'|(\Var[v\cdot S']+(v\cdot (\mu_T-\mu_{S'}))^2).
\end{equation}
By the second stability condition, we have that $|\Var[v\cdot S'] -1| \leq \delta^2/\eps$.
Furthermore, the stability condition and {Lemma \ref{lem:stability}} give
$$\|\mu_T-\mu_{S'}\|_2 \leq \|\mu_T-\mu_X\|_2+\|\mu_X-\mu_{S'}\|_2 = O(\delta+\sqrt{\eps \lambda})\;.$$
Since $\lambda \geq 8\delta^2/\eps$, combining the above gives that
$\sum_{x\in T\setminus S} g(x) \geq (2/3)|S|\lambda$.
Moreover, since $|L|\geq |T \setminus S|$ and since $g$ takes its largest values on points $x \in L$,
we have that
$$
\sum_{x\in T}f(x) = \sum_{x\in L}g(x) \geq \sum_{x\in T \setminus S}g(x) \geq (16/3) |S|\delta^2/\eps \;.
$$
Comparing the results of Equation \eqref{subsetGSumEqn} with $S'=S$ and $S'=S \setminus L$, we find that
\begin{align*}
\sum_{x\in S\cap T}f(x) & = \sum_{x\in S\cap L} g(x) = \sum_{x\in S}g(x) - \sum_{x\in S\setminus L}g(x)\\
& = |S|(1\pm \delta^2/\epsilon + O(\delta^2+\epsilon\lambda))-|S\setminus L|(1\pm \delta^2/\epsilon + O(\delta^2+\epsilon\lambda))\\
& \leq 2|S|\delta^2/\epsilon + |S|O(\delta^2+\epsilon\lambda).
\end{align*}
The latter quantity is at most $(1/2) \sum_{x\in T}f(x)$
when $\delta$ and $\eps/\delta$ are sufficiently small constants.
This completes the proof of Proposition~\ref{prop:univ-filter}.
\end{proof}

\subsection{Bibliographic Notes} \label{ssec:lit-robust-mean}
The convex programming and filtering methods described in this article
appeared in~\cite{DKKLMS16, DKK+17}. Here we gave a simplified and 
unified presentation of these techniques.
The idea of removing
outliers by projecting on the top eigenvector of the empirical covariance
goes back to~\cite{KLS09}, who used it in the context of learning linear separators with malicious noise.
That work \cite{KLS09} used a ``hard'' filtering step which only removes outliers,
and consequently leads to errors that scale logarithmically with the dimension. 
Subsequently, the work of \cite{ABL17} employed a soft-outlier
removal step in the same supervised setting as~\cite{KLS09}, to obtain improved bounds for that problem.
It should be noted that the soft-outlier method of \cite{ABL17} is similarly insufficient
to obtain dimension-independent error bounds for the unsupervised setting.

The work of~\cite{LaiRV16} developed a recursive dimension-halving technique
for robust mean estimation. Their technique leads to error $O(\eps \sqrt{\log(1/\eps)} \sqrt{\log d})$
for Gaussian robust mean estimation in Huber's contamination model.
In short, the algorithm of \cite{LaiRV16} begins by removing any extreme outliers from the input set of $\eps$-corrupted samples. This ensures that, after this basic outlier removal step, the empirical covariance matrix has trace $d(1+\tilde O(\eps))$, which implies that the $d/2$ smallest eigenvalues are all at most $1+\tilde O(\eps)$. This allows \cite{LaiRV16} to show, using techniques akin to Lemma \ref{lem:stability}, that the projections of the true mean and the empirical mean onto the subspace spanned by the corresponding (small) eigenvectors are close. The \cite{LaiRV16} algorithm then uses this approximation
for this projection of the mean, projects the remaining points onto the orthogonal subspace,
and recursively finds the mean of the other projection.

In addition to robust mean and covariance estimation,~\cite{DKKLMS16, LaiRV16} gave robust learning
algorithms for various other statistical tasks, including robust density estimation for mixtures of spherical
Gaussians and binary product distributions, robust independent component analysis (ICA), and robust
singular value decomposition (SVD). Building on the robust mean estimation techniques of~\cite{DKKLMS16},
~\cite{ChengDKS18} gave robust parameter estimation algorithms for Bayesian networks
with known graph structure. Another extension of these results was found by \cite{SteinhardtCV18}
who gave an efficient algorithm for robust mean estimation with respect to all $\ell_p$-norms.

The algorithmic approaches described in this section
robustly estimate the mean of a spherical Gaussian within error $O(\eps \sqrt{\log(1/\eps)})$
in the strong contamination model of Definition~\ref{def:adv}. A more sophisticated filtering technique that achieves
the optimal error of $O(\eps)$ in the additive contamination model was developed in
\cite{DiakonikolasKKLMS18}. Very roughly, this algorithm proceeds by using a novel filter
to remove bad points if the empirical covariance matrix has {\em many} eigenvalues of size $1+\Omega(\eps)$.
Otherwise, the algorithm uses the empirical mean to estimate the mean on the space spanned by small eigenvectors,
and then uses brute force to estimate the projection onto the few principal eigenvectors.
For the strong contamination model, it was shown in~\cite{DKS17-sq} that any improvement
on the $O(\eps \sqrt{\log(1/\eps)})$ error requires super-polynomial time in the Statistical Query model.

Finally, we note that ideas from~\cite{DKKLMS16} have led to
proof-of-concept improvements in the analysis of genetic data~\cite{DKK+17} and in
adversarial machine learning~\cite{DiakonikolasKKLSS2018sever, TranLM18}.

\section{Beyond Robust Mean Estimation} \label{sec:gen}
In this section, we provide an overview of the ideas behind
recently developed robust estimators for more general statistical tasks.
This section follows the structure of a STOC 2019 tutorial by the authors~\cite{DK19-stoc-tutorial}.

\subsection{Robust Stochastic Optimization} \label{ssec:stoch-opt}


It turns out that the algorithmic techniques for high-dimensional robust mean estimation
described in the previous section can be viewed as useful primitives
for robustly solving a range of machine learning problems.
More specifically, we will argue in this section that any efficient robust mean estimator
can be used (in essentially a black-box manner) to obtain efficient robust algorithms
for machine learning tasks that can be expressed as stochastic optimization problems.


In a stochastic optimization problem, we are given samples from an unknown
distribution $\mathcal{F}$ over functions $f: \R^d \to \R$, and our goal is to
find an approximate minimizer of the function $F(w) = \E_{f \sim \mathcal{F}}[f(w)]$
over $\mathcal{W} \subseteq \R^d$.
This framework encapsulates a number of well-studied machine learning problems.
First, we note that the problem of mean estimation can be expressed in this form,
by observing that the mean of a distribution $X$ is the value
$\mu = \arg\min_{w \in \R^d} \E_{x \sim X}[\|w-x\|_2^2]$. That is, given a sample $x \sim X$,
the distribution $\mathcal{F}$ over functions $f_x(w)=\|w-x\|_2^2$ turns
the task of mean estimation into a stochastic optimization problem.
A more interesting example is the problem of least-squares linear regression:
Given a distribution $\D$ over pairs $(x, y)$, where $x \in \R^d$ and $y \in \R$, we want
to find a vector $w \in \R^d$ minimizing $\E_{(x, y) \sim \D}[(w \cdot x - y)^2]$.
Similarly, the problem of linear regression fits in the stochastic optimization framework
by defining the distribution $\mathcal{F}$ over functions
$f_{(x,y)} (w) = (w \cdot x-y^2)$, where  $(x,y) \sim \D$. Similar formulations exist
for numerous other machine learning problems, including $L_1$-regression, logistic regression,
support vector machines, and generalized linear models (see, e.g.,~\cite{DiakonikolasKKLSS2018sever}).
Finally, we note that the stochastic optimization framework encompasses non-convex problems as well.
For example, the general and challenging problem of training a neural net
can be expressed in this framework, where $w$ represents some high-dimensional vector of parameters classifying
the net and each function $f(w)$ quantifies how well that particular net classifies a given data point.


Before we discuss {\em robust} stochastic optimization, we make a few basic remarks
regarding the non-robust setting. We start by noting that, without any assumptions, the problem of optimizing
the function $F(w) = \E_{f \sim \mathcal{F}}[f(w)]$, even approximately, is NP-hard. On the other hand,
in many situations, it suffices to find an approximate critical point of $F$, i.e., a point $w$ such that
$\|\nabla F(w)\|_2$ is small. For example, if $F$ is convex (which holds if each $f \sim \mathcal{F}$ is convex),
an approximate critical point is also an approximate global minimum. For several structured non-convex problems,
an approximate critical point is also considered a satisfactory solution.
On input a set of clean samples, i.e., an i.i.d. set of functions $f_1, \ldots, f_n \sim \mathcal{F}$,
we can efficiently find an approximate critical point of $F$ using (projected) gradient descent.
For more structured problems, e.g., linear regression, faster and more direct methods may be available.


In the robust setting, we have access to an $\eps$-corrupted training
set of functions $f_1, \ldots, f_n$ from $\mathcal{F}$.
Unfortunately, even a {\em single} corrupted sample can completely compromise
the guarantees of gradient descent. The robust version of this problem was first studied by \cite{CSV17}, who considered the problem in the case where a majority of the datapoints are outliers. The vanilla
outlier-robust setting, where $\eps<1/2$, was first studied in two
concurrent works~\cite{PrasadSBR2018, DiakonikolasKKLSS2018sever}.
The main intuition present in both these works
is that robustly estimating the gradient of the objective function can be viewed
as a robust mean estimation problem. As a result, if an efficient robust gradient oracle
is available, we can ``simulate'' gradient descent and compute an approximate critical
point of $F$. Note that this method employs a robust mean estimation algorithm
at every step of gradient descent.


The work of~\cite{DiakonikolasKKLSS2018sever} also proposed an alternative approach, which
turns out to be much faster in practice.
Instead of using a robust gradient estimator as a black-box, one uses {\em any} approximate
empirical risk minimizer (ERM) in conjunction with the filtering algorithm for robust mean estimation
of the previous section. This method only requires black-box access to an approximate ERM and
calls the filtering routine only when the ERM reaches an approximate critical point.
The correctness of this algorithm relies on structural properties of the filtering method.
Roughly speaking, the main idea is as follows:
Suppose that we have reached an approximate critical point $w$
of the empirical risk and at this stage we apply a filtering step.
By the guarantees of the filter, we know that we are in one of two cases:
either the filtering step removes more outliers (corrupted functions) than inliers (in expectation),
or it certifies that the gradient of $F$ at $w$ is close to the gradient of the empirical risk at $w$.
In the former case, we make progress as we produce a ``cleaner'' set of functions. In the latter case,
since $w$ is an approximate critical point of the empirical risk, our certificate implies that
$w$ is also an approximate critical point of $F$, as desired.


For the above robust optimization
approaches to be computationally efficient, we require some assumptions
on the distribution of the clean samples (functions). With no such assumption,
most of the size of $F$ could be determined by a small fraction of the $f$'s in such a way that it is
computationally intractable to determine whether these values are due to corruptions or not.
A natural condition used in~\cite{CSV17, DiakonikolasKKLSS2018sever}
is that for every $w$ the covariance matrix of the gradient distribution, $\Cov_{f\sim \mathcal{F}}[\nabla f(w)]$,
is bounded from above. Under this condition, if one has enough samples from $\mathcal{F}$
(so that the empirical covariance of the clean samples is bounded for all $w$),
one can use the filtering algorithm of Section~\ref{ssec:filter}
to robustly estimate $\nabla F(w)$ to $\ell_2$-error $O(\sqrt{\eps})$ for any $w$.
Using either of the two afore-described approaches, one can find a point $w$ such that
$\|\nabla F(w)\|_2 = O(\sqrt{\eps})$~\cite{DiakonikolasKKLSS2018sever}. We note that~\cite{PrasadSBR2018}
uses the robust mean estimator of~\cite{LaiRV16}, which requires somewhat stronger distributional
assumptions and incurs error scaling logarithmically with the dimension.

In summary, we have described two meta-algorithms for robust stochastic optimization.
An interesting open problem is to obtain a faster algorithm for this general task, in particular
one that has information-theoretically optimal sample complexity and
uses a minimum number of queries to an ERM oracle.

\medskip


\noindent {\em Robust Linear Regression.}
As already mentioned, if $F$ is convex, the approximate critical point computed via
robust stochastic optimization translates to an approximate
global minimizer. In the following paragraphs, we describe how to obtain
an approximate global minimum for the fundamental task of linear regression.
Several other applications are given in~\cite{DiakonikolasKKLSS2018sever, PrasadSBR2018}.


We will focus on the following standard setup:
We are given a collection of labeled examples $(x^{(i)}, y^{(i)})$,
where $x^{(i)}$ is drawn from a distribution $X$ on $\R^d$
and $y^{(i)}=\beta \cdot x^{(i)}+e^{(i)}$, where $e^{(i)}$ is drawn from
a distribution $e$ that is independent of $X$, and has mean $0$ and variance $1$.
The objective is to find a vector $w^{\ast} \in \R^d$ that approximately minimizes the function $F(w) = \E_{(x, y)}[f_{(x,y)}(w)]$,
where $f_{(x,y)}(w) = (w \cdot x - y)^2$.

Recall that the robust stochastic optimization approach described in the previous
paragraphs relies on the assumption that the covariance matrix of the gradients
$\Cov_{f\sim \mathcal{F}}[\nabla f(w)]$ is bounded. This condition translates
to certain necessary conditions about the distribution $X$.
To better understand these conditions, we first consider the gradient of $f$.
We have that
$$\nabla_w f_{(x,y)}(w) = 2 (w\cdot x -y) x  = 2((w-\beta) \cdot x - e) x \;.$$
Using this expression, it is not hard to see that the variance of the gradient $\nabla_w f_{(x,y)}(w)$
in the $v$ direction is equal to
$$4\E_{x \sim X}[(v\cdot x)^2] + 4 \E_{x \sim X}[(v\cdot x)^2((w-\beta)\cdot x)^2] \;.$$
Note that the first term above is bounded, as long as $X$
has bounded covariance. To bound the second, we will need to assume
that $X$ has bounded fourth moments. In particular, if we assume that
$X$ has $\E_{x \sim X}[(v\cdot x)^4]=O(1)$ for all unit vectors $v$,
the covariance of the gradients has maximum eigenvalue
bounded by $O(1+\|w-\beta\|_2^2)$.
For simplicity, let us assume that we know a priori a ball of constant $\ell_2$-radius containing $\beta$.
Then, using our robust stochastic optimization routine,
we can efficiently compute an approximate critical point with gradient of $\ell_2$-norm
$O(\sqrt{\eps})$.

We now show that such an approximate critical point is also an approximate global minimum.
Note that the gradient of $F$ at $w$ equals $\E[2((w-\beta)\cdot x - e) x] = 2 \E[XX^T] (w-\beta).$
For convenience, we will additionally assume that $\E[XX^T]$ is bounded from below
by a constant multiple of the identity matrix. This means that for any $w \in \R^d$
we have that $\|w-\beta\|_2 = O(\|\nabla_w F(w)|\|_2)$.
Therefore, an approximate critical point of $F$ is equivalent to
a good approximation of $\beta$, which is the global minimizer of $F$.

The above immediately gives an approximate global minimizer with $\ell_2$-error $O(\sqrt{\eps})$, assuming
we started with a constant radius ball containing $\beta$. It is not difficult
to handle a very rough approximation to $\beta$ with error at most $R$.
A simple (but somewhat inefficient) method is to search
for $w$ within a ball of radius $R$ in which the covariance
of the gradients is bounded by $O(1+R^2)$. For $R>1$, this
guarantees a new point which approximates $\beta$ within $\ell_2$-error $O(\sqrt{\eps}R)$.
Iterating this procedure, we can achieve a final error of $O(\sqrt{\eps})$.
A detailed and more efficient procedure is described in~\cite{DiakonikolasKKLSS2018sever}.


Finally, it should be noted that the problem of robust linear regression has been extensively studied
in recent years. Using the Sums-of-Squares hierarchy, \cite{KlivansKM18} developed computationally efficient algorithms for robust linear regression. For the case where the covariates follow a Gaussian distribution, \cite{DKS19-lr} obtained
computationally efficient algorithms with near-minimax sample complexity and error guarantee.
There has also been recent related work~\cite{BhatiaJK15, BhatiaJKK17, SuggalaBR019},
which proposed efficient algorithms for ``robust'' linear regression
in a restrictive corruption model that only allows adversarial corruptions to the responses, but not to the covariates.

\subsection{Robust Covariance Estimation} \label{ssec:cov}

The algorithmic techniques for high-dimensional robust mean estimation
described in Section~\ref{sec:robust-mean} can be generalized to robustly estimate
higher moments under appropriate assumptions.
In this section, we describe how to adapt the filter technique to
robustly estimate the covariance matrix $\Sigma$ of a distribution $X$
satisfying appropriate moment conditions.


First note that we can assume without loss of generality that $X$ is centered.
Specifically, by considering the differences between pairs of $\eps$-corrupted
samples from $X$, we have access to a set of $2\eps$-corrupted
samples from a distribution $X'$ with mean $0$ and covariance matrix $2 \Sigma$.


The basic idea underlying this section is fairly simple: Robustly estimating
the covariance matrix of a centered random variable $X$
is essentially equivalent to robustly estimating the mean of the random variable
$Y=XX^T.$ 
That is, the problem of robust covariance estimation can be ``reduced'' to the problem
of robust mean estimation of a more complicated random variable.
If the random variable $Y$ satisfies appropriate moment conditions (or tail bounds),
we can hope to apply the techniques of the previous section.
This is the approach taken by~\cite{DKKLMS16, LaiRV16}.
At a very high-level, it is possible to design a robust covariance estimation algorithm using
the filtering method~\cite{DKKLMS16}, where each filtering step
removes points based on the empirical fourth moment tensor.


Formalizing the above approach requires some care for the following reason:
The robust mean estimation techniques for a distribution $Y$
require an {\em a priori upper bound} on its covariance $\Cov[Y]$.
Unfortunately, such bounds do not hold for our random variable $Y = X X^T$, even if
$X$ is a Gaussian distribution. To handle this issue, we need to use additional structural
properties of $X$. Specifically, if $X \sim \mathcal{N}(0, \Sigma)$,
we can leverage the fact that the covariance of $Y$ can be expressed as a function of the covariance of $X$.
An upper bound on $\Sigma$ will give us an upper bound
on the covariance of $Y$, which can then be used to obtain a better approximation of $\Sigma$.
Applying this idea iteratively will allow us to bootstrap better and better approximations,
until we end up with an approximation to $\Sigma$ with error close to the information-theoretic optimum.
This method is proposed and analyzed (with slightly different terminology)
in~\cite{DKKLMS16, DiakonikolasKKLMS18}.

Before we proceed with a detailed outline of the method, we should clarify the metric
we will use to approximate the covariance matrix. Recall that for mean estimation we
used the $\ell_2$-norm between vectors. A natural choice for the covariance is the Frobenius
norm, i.e., we would like to find an estimate $\wh{\Sigma}$ of the true matrix $\Sigma$
such that $\| \wh{\Sigma}- \Sigma \|_F$ is small.
Here we will use the Mahalanobis distance, which is affine invariant and intuitively
corresponds to multiplicative approximation. Specifically, we want to compute an estimate  $\wh{\Sigma}$ such that
$\| \Sigma^{-1/2} \wh{\Sigma} \Sigma^{-1/2}- I \|_F$ is small. A basic fact motivating the use of this metric is
that the total variation distance
between two Gaussian distributions $\mathcal{N}(0, \Sigma)$ and $\mathcal{N}(0, \wh{\Sigma})$
is bounded from above by $O(\| \Sigma^{-1/2} \wh{\Sigma} \Sigma^{-1/2}- I \|_F)$.

For the remainder of this section, we will assume for concreteness that $X \sim \mathcal{N}(0, \Sigma)$ and
we will describe an efficient algorithm for robust covariance estimation that achieves Mahalanobis distance
$O(\eps \log(1/\eps))$, which is within a logarithmic factor from the information-theoretic optimum of
$\Theta(\eps)$. We note that the same approach gives error $O(\sqrt{\eps})$ for any distribution
whose fourth moment tensor is appropriately bounded by a function of the covariance.


To get started, we first need to understand the relationship between
$\mathbf{\Sigma} := \Cov[Y]$ and $\Sigma$. To that end, let us denote by $A^{\mathrm{flat}}$
the canonical flattening of a matrix $A$ into a vector. With this notation, it is not hard to verify that
\begin{equation}\label{boldSigmaEqn}
A^{\mathrm{flat}}\mathbf{\Sigma}A^{\mathrm{flat}} = 2\left \|\Sigma^{1/2}\left(\frac{A+A^T}{2} \right)\Sigma^{1/2} \right\|_F^2 \;.
\end{equation}
This formula essentially expresses $\mathbf{\Sigma}$ in terms of the quadratic form that it defines.
An important consequence of Equation \eqref{boldSigmaEqn} is the following:
Given covariance matrices $\Sigma,\Sigma'$ and the corresponding matrices
$\mathbf{\Sigma},\mathbf{\Sigma'}$, we have that if $\Sigma \preceq \Sigma'$,
then $\mathbf{\Sigma} \preceq \mathbf{\Sigma'}$. In other words, if we have
an upper bound on the true covariance $\Sigma$, this gives us an upper bound on the covariance of $Y$.

Specifically, if $\Sigma \preceq \Sigma_0$, for some matrix $\Sigma_0$, we have that
$\Cov[\Sigma_0^{-1/2}Y\Sigma_0^{-1/2}] = O(I)$.
Using our robust mean estimator for random variables with bounded covariance
will allow us to approximate $\E[\Sigma_0^{-1/2} Y\Sigma_0^{-1/2}] = \Sigma_0^{-1/2}\Sigma \Sigma_0^{-1/2}$
to error $O(\sqrt{\eps})$ in Frobenius norm. This gives us an estimate $\wh{\Sigma}$ such that
$\| \Sigma_0^{-1/2} (\wh{\Sigma} -\Sigma) \Sigma_0^{-1/2} \|_F  = O(\sqrt{\eps})$. This means
that, given an upper bound $\Sigma_0$ on $\Sigma$, we can obtain a better one.
To obtain an initial upper bound $\Sigma_0$, we note that twice the sample covariance
of a large set of samples from $X$ provides an upper bound on the true covariance of $X$ even with corruptions,
as although corruptions can substantially increase the empirical covariance of $X$, they cannot
decrease it by much. Starting from $\Sigma_0$, we obtain a new approximation
$\Sigma_1=\Sigma+O(\sqrt{\eps})\Sigma_0$; from this we can obtain an improved approximation
$\Sigma_2=\Sigma+O(\sqrt{\eps})\Sigma_1$, and so forth. Iterating this technique
yields a matrix $\wh{\Sigma}$ such that $\|\Sigma^{-1/2}(\wh{\Sigma} - \Sigma)\Sigma^{-1/2}\|_F =O(\sqrt{\eps})$.


The error guarantee of $O(\sqrt{\eps})$ achieved above is already fairly accurate.
Once we have such a good approximation to the true covariance $\Sigma$,
we can improve the error guarantee even further by using stronger tail bounds for the Gaussian distribution.
In the following, we will assume a rough scale for $\Sigma$, namely that $I \preceq \Sigma \preceq 2I$.
Suppose that, for some $\delta>\eps$, we have a matrix $\Sigma_0$ satisfying $\|\Sigma_0-\Sigma\|_F \leq \delta$.
We can then use Equation \eqref{boldSigmaEqn} to approximate $\mathbf{\Sigma}$ from $\Sigma_0$.
It is not hard to see that if we know $\Sigma$ within Frobenius norm $\delta$, this allows us to
compute $\mathbf{\Sigma}$ within spectral norm $O(\delta)$. Thus, after applying an appropriate
linear transformation to $Y$, we obtain \new{a random variable $Y'$} with covariance within $O(\delta)$
of the identity matrix in spectral norm. We claim that \new{$Y'$} has strong tail bounds.
This follows from standard tail bounds for degree-$2$ polynomials over Gaussian random variables.
Specifically, for any unit vector $v$, $v\cdot \new{Y'}$ is a quadratic polynomial in $X$
with variance $O(1)$. Standard results imply that $v\cdot \new{Y'}$ has exponential tails.
From this we can show that any sufficiently large number of samples from $Y'$
will be $(O(\sqrt{\eps\delta}+\eps\log(1/\eps)),\eps)$-stable with high probability.
In summary, if we know $\Sigma$ to Frobenius error $\delta$,
we can use robust mean estimation techniques to learn
it to Mahalanobis error $O(\sqrt{\eps\delta}+\eps\log(1/\eps))$.
Iterating this, we can obtain a final Mahalanobis  error of $O(\eps\log(1/\epsilon))$, as desired.


We conclude by noting that the aforementioned can be used to robustly learn a
Gaussian distribution with unknown mean and covariance as follows:
First, one learns the covariance as above (by reducing to the mean $0$ case).
Then one learns the mean of a Gaussian with an (approximately) known covariance matrix.
In summary, one obtains a hypothesis Gaussian $\mathcal{N}(\wh{\mu}, \wh{\Sigma})$ within total variation
distance $O(\eps\log(1/\eps))$ from the uncorrupted distribution $\mathcal{N}(\mu, \Sigma)$.

\subsection{Robust Sparse Estimation Tasks} \label{ssec:sparse}
The task of leveraging sparsity in high-dimensional parameter estimation
is a well-studied problem in statistics. In the context of robust estimation,
this problem was first considered in~\cite{BDLS17}, which adapted
the unknown convex programming method of~\cite{DKKLMS16} described
in this article. Here we focus on robust sparse mean estimation
and describe two algorithms: the convex programming algorithm
of~\cite{BDLS17} and a novel filtering method~\cite{DKKPS19-sparse} that only
uses spectral operations.

Formally, given $\eps$-corrupted samples from $\mathcal{N}(\mu, I)$,
where the mean $\mu$ is unknown and assumed to be $k$-sparse,
i.e., supported on an unknown set of $k$ coordinates,
we would like to approximate $\mu$, in $\ell_2$-distance. Without corruptions,
this problem is easy: We draw $O(k\log(d/k)/\eps^2)$ samples and output
the empirical mean truncated in its largest magnitude $k$ entries.
The goal is to obtain similar sample complexity and error guarantees
in the robust setting.

At a high level, we note that the truncated sample mean should be accurate
as long as there is no $k$-sparse direction in which the error between
the true mean and sample mean is large. This condition can be certified,
as long as we know that the sample variance of $v\cdot X$ is close to
$1$ for all unit $k$-sparse vectors $v$. This would in turn allow us
to create a filter-based algorithm for $k$-sparse robust mean estimation that uses
only $O(k\log(d/k)/\eps^2)$ samples. While this idea naturally leads to a sample-optimal
robust algorithm for the problem, it is computationally infeasible.
This holds because the problem of determining whether there is a $k$-sparse direction with large
variance (sparse PCA) is known to be computationally hard,
even under natural distributional assumptions~\cite{BerthetR13}.
To circumvent this hardness result,~\cite{BDLS17} considers a convex relaxation of sparse PCA,
which leads to a polynomial-time version of the aforementioned algorithm
that requires $O(k^2 \log(d/k)/\eps^2)$ samples. Moreover, {there is} evidence~\cite{DKS17-sq},
in the form of a lower bound in the Statistical Query model (a restricted but powerful computational model),
that this quadratic blow-up in the sample complexity is necessary for polynomial-time algorithms.
Note that although the $O(k^2 \log(d/k)/\eps^2)$ sample complexity is worse than the information-theoretic
optimum of $\Theta(k\log(d/k)/\eps^2)$, for small $k$, it is still substantially better than
the $\Omega(d/\eps^2)$ sample size required by dense methods.

The convex-programming algorithm of~\cite{BDLS17} works as follows:
Let $\wh{\Sigma}$ be the empirical covariance matrix.
If there is a $k$-sparse unit vector $v$ with $v^T \wh{\Sigma} v$ large,
we have that $\tr(\wh{\Sigma} vv^T)$ is large. Here $vv^T$ is a positive semi-definite,
trace-$1$ matrix whose entries have $\ell_2$-norm at most $1$ and
$\ell_1$-norm at most $k$ (the latter following from the sparsity of $v$).
The work of~\cite{BDLS17} considers the following convex relaxation
of the problem of finding the sparse vector $v$:
Find a positive semi-definite, trace-$1$ matrix $H$, whose entries have $\ell_2$-norm
at most $1$ and $\ell_1$-norm at most $k$, so that $\tr(\wh{\Sigma} H)$ is as large as possible.
If the optimal solution to this convex relaxation is small, we have certified that
$\wh{\Sigma}$ has no sparse directions of large variance,
and consequently we have certified the accuracy of the truncated empirical mean.
On the other hand, if the optimal value of the convex relaxation is large, we have found a
``sparse direction'' of large variance and can use this to refine our set of samples. In particular,
\cite{BDLS17} use the ellipsoid method to find a subset of the samples
so that for all such $H$, $\tr(\wh{\Sigma} H)$ is not too large.

To bound the sample complexity of this method, one needs to show that
with a sufficiently large set $S$ of iid samples from $\mathcal{N}(\mu, I)$
the empirical covariance of $S$, $\wh{\Sigma}_S$, satisfies that
the trace $\tr(\wh{\Sigma}_S H)$ is appropriately small for all such $H$. 
One can show this as follows:
First, is is easy to see that, given a set $S$ of size $\Omega(k^2 \log(d)/\eps^2)$,
every entry of $\wh{\Sigma}_S - I$ is $O(\eps/k)$ with high probability. If this holds, then we have
that
$$
\tr(\wh{\Sigma}_S H) = \tr(H)+\tr((\wh{\Sigma}_S-I)H) = 1 + O(\eps/k)\cdot k = 1 + O(\eps) \;,
$$
where the second equality is because the entries of $\wh{\Sigma}_S-I$ have bounded
$\ell_{\infty}$-norm and the entries of $H$ have bounded $\ell_1$-norm.
Therefore, with $\Omega(k^2\log(d)/\eps^2)$ samples, this algorithm can be shown
to work with high probability.

In subsequent work, \cite{LSLC18} gave an iterative filter-based method
for robust sparse mean estimation, which avoids the use of the ellipsoid method
but still requires multiple solutions to the convex relaxation of sparse PCA in each filtering iteration.
Another algorithm for robust sparse mean estimation, proposed by
~\cite{LLC19}, works via iterative trimmed thresholding.
While this algorithm seems practically viable in terms of runtime,
it can tolerate vanishingly small fraction of outliers.

More recently, \cite{DKKPS19-sparse} developed iterative spectral algorithms
for robust sparse estimation tasks (including sparse mean estimation and sparse PCA).
These algorithms achieve the same error guarantees as~\cite{BDLS17}, while being
significantly faster. In the context or robust sparse mean estimation, the algorithm
of \cite{DKKPS19-sparse} considers the $O(k^2)$ largest entries of $\wh{\Sigma}-I$.
If the $\ell_2$-norm of these entries is much larger than $\eps$, it follows that
there is a sparse, degree-$2$ polynomial $p(x)$ where the expectation of
$p$ over all samples in $S$ is substantially different from its average value over the clean samples.
This allows us to build a filter for points based on their values under $p$.
On the other hand, if this is not the case, it means that, for sparse vectors $v$,
the contribution to $v^T(\wh{\Sigma} - I)v$ coming from entries other
than the top $O(k^2)$ ones is small. Therefore, $v^T \wh{\Sigma} v$ would
only be large if we could find such a $v$ supported only on the rows and columns of these $O(k^2)$ entries.
We can then check for all vectors $v$ on this limited support. A careful analysis
shows that with $\tilde O(k^2\log(d)/\eps^2)$ samples
the appropriate concentration conditions hold
for every $k^2$-sparse vector and degree-$2$
polynomial. This allows the appropriate filters to work with this sample size.


\subsection{List-Decodable Learning} \label{ssec:list-dec}


In this article, we focused on the classical robust statistics setting, where the outliers
constitute the minority of the dataset, quantified by the proportion of contamination
$\eps<1/2$, and the goal is to obtain estimators with error scaling
as a function of $\eps$ (and is independent of the dimension $d$).
A related setting of interest focuses on the regime that the fraction $\alpha$
of clean data (inliers) is small -- strictly smaller than $1/2$. That is, we observe
$n$ samples, an $\alpha$-fraction of which (for some $\alpha <1/2$)
are drawn from the distribution of interest, but the rest are arbitrary.

This question was first studied in the context of mean estimation in~\cite{CSV17}.
A first observation is that, in this regime, it is information-theoretically impossible
to estimate the mean with a single hypothesis. Indeed, an adversary can
produce $\Omega(1/\alpha)$ clusters of points each drawn from a
good distribution with different mean. Even if the algorithm could
learn the distribution of the samples exactly, it still would not be able to
identify which of the clusters is the correct one. To circumvent this bottleneck,
the definition of learning must be somewhat relaxed. In particular,
the algorithm should be allowed to return {\em a small list of hypotheses} with the
guarantee that \emph{at least one} of the hypotheses is close to the true mean.
This is the model of {\em list-decodable learning}, a learning model introduced by~\cite{BBV08}.
Another qualitative difference with the small $\eps$ regime
is that in list-decodable learning, it is often information-theoretically
necessary for the error to increase without bound 
as the fraction of clean data $\alpha$ goes to $0$. In summary,
given polynomially many corrupted samples, we would like to output $O(1/\alpha)$
(or $\mathrm{poly}(1/\alpha)$) many hypotheses
with the guarantee that (with high probability) at least one hypothesis
is within $f(\alpha)$ of the true mean, where $f(\alpha)$ depends
on the concentration properties of the distribution in question,
but otherwise is information-theoretically best possible.

The information-theoretic limits of list-decodable mean estimation
have only been addressed very recently. The work \cite{DiakonikolasKS18-mixtures} gave
nearly tight bounds on the minimum error achievable for list-decodable mean estimation on $\R^d$
(with $\poly(1/\alpha)$ candidate hypotheses) for structured distribution families,
including Gaussians and distributions with bounded covariance. In particular, the optimal $\ell_2$-error
was determined to be $\Theta(\sqrt{\log(1/\alpha)})$
for spherical Gaussians and $\Theta(\alpha^{-1/2})$ for bounded covariance distributions.
The algorithmic aspects of list-decodable mean estimation have turned out to be much more challenging.
For bounded covariance distributions,~\cite{CSV17} gave an SDP-based algorithm achieving
near-optimal error of $\tilde{O}(\alpha^{-1/2})$. In the rest of this section, we describe
a generalization of the filtering method for list-decodable mean estimation
introduced in~\cite{DiakonikolasKS18-mixtures}.

\paragraph{List-Decodable Mean Estimation via (Multi-) Filters.}
The filtering techniques discussed in Section~\ref{ssec:filter}
can be adapted to work in the list-decodable setting as well. For the remainder of this discussion,
we will restrict ourselves to list-decodable mean estimation when the clean data
is drawn from an identity covariance Gaussian distribution. At a high-level,
the adaptation of the filtering method works as follows:
If the sample covariance matrix has no large eigenvalues, this certifies that the true mean
and sample mean are not too far apart. However, if a large eigenvalue exists, the construction of a
filter is more elaborate. To some extent, this is a necessary difficulty because the algorithm must
return multiple hypotheses. To handle this issue, one needs to construct a {\em multi-filter},
which may return {\em several subsets} of the original dataset with the guarantee
that at least one of them is cleaner than the original. Such a multi-filter
was first introduced in~\cite{DiakonikolasKS18-mixtures}.

We now proceed with a detailed overview.
The main idea is to employ some type of filtering to obtain a subset $S$
of our original dataset \new{$T$} so that the following conditions are satisfied:
(i) The set $S$ contains at least half of the original clean samples \new{in $T$};
and (ii) The empirical covariance of $S$ is bounded from above
by some small parameter $\sigma>0$ in every direction.
If such a subset $S$ can be efficiently computed,
we can \new{certify} that the empirical mean of $S$
will be close to the true mean. To show this, let $\mu, \mu_G$, and $\mu_S$
denote the true mean of the uncorrupted distribution, the mean of the clean samples in $S$
and the mean of all the samples in $S$, respectively. By condition (i), it is easy
to see that $\|\mu-\mu_G\|_2 = O(1)$. On the other hand, the variance of $S$ in the
$\mu_G-\mu_S$ direction is at least $(\alpha/2)\|\mu_G-\mu_S\|_2^2$,
since an at least $(\alpha/2)$-fraction of clean samples have distance $\|\mu_G-\mu_S\|_2$
from the full mean. Combining with condition (ii), we have that $\|\mu_G-\mu_S\|_2=O(\sqrt{\sigma/\alpha})$,
and hence by the triangle inequality it follows that
$\|\mu-\mu_S\|_2 = O(1+\sqrt{\sigma/\alpha})$.


We can attempt to use the filtering approach of Section~\ref{ssec:filter} to
find such a set $S$. If the initial set of samples \new{$T$} has bounded covariance, then its
empirical mean works, so we can use it as our set $S$.
Otherwise, we can project the samples \new{in $T$} on a direction of large variance
in an attempt to remove outliers. Unfortunately, the outlier removal step
cannot be so straightforward in this setting.
The filtering steps we have described so far generally work by first deriving
an approximation to the true mean and then removing samples
that are too far away from it. However, the first step of this procedure inherently
fails here, since the outliers constitute the majority of the dataset.
In particular, if the initial set of samples \new{$T$} come in two large but separated clusters,
we will not be able to determine which cluster contains the true mean,
and thus will not be able to find any points that are definitively outliers.
This difficulty is of course necessary, as the list-decodable algorithm
is in general required to produce several hypotheses. To circumvent this issue,
our algorithm will return {\em multiple subsets} of points with the guarantee
that at least one of these subsets is cleaner than the original. We will call
such an algorithm a {\em multi-filter}.

Given a set $S$ of samples containing at least half
of the original clean points and a direction in which $S$ has large variance,
we want our multi-filter to return a collection of (potentially overlapping)
subsets $S_i$ of $S$ with the following properties:  First, we need it to be the case that
for \emph{at least one} of these sets, at most an $\alpha/2$-fraction
of the points in $S \setminus S_i$ are clean. Second, we need to ensure that the blowup
in the number of such subsets is not too large. One way to achieve
this is to require that $\sum_i |S_i|^2 \leq |S|^2$.

Our overall algorithm will work by maintaining several sets $S_i$ of samples.
If any of these sets has too large variance in some direction,
we will apply the multi-filter replacing it with several smaller subsets.
We note that by the first condition above, if we started with a set
where at most an $\alpha/2$-fraction of its complement was clean,
at least one of the subsets will also have this property. Therefore,
at the end of this procedure, we are guaranteed to end up with
at least one $S_i$ satisfying the conditions necessary
to give a good approximation to the true mean. Our second condition
will imply that at any stage of this algorithm, the sum of the squares
of the sizes of the $S_i$'s will never exceed the squared size of our original set of samples.
This condition guarantees that the sample complexity and runtime of the overall algorithm are polynomial.
In fact, by observing that we only need to return the sample mean of $S_i$'s
that contain at least an $\alpha/2$-fraction of our original set of samples,
we will have at most $O(1/\alpha^2)$ hypotheses.
To reduce the list size of returned hypotheses further, there is a simple method
\cite{DiakonikolasKS18-mixtures} that shows how to efficiently
reduce any set of polynomially many hypotheses (at least one of which is guaranteed
to be within $r$ of the true mean) to a list of $O(1/\alpha)$ hypotheses
at least one of which is nearly as close.

The multi-filter step works as follows: Given a direction $v$
in which the variance is too large, there are two ways we can attempt
to get an appropriate collection of subsets $S_i$. First, if there is some
interval $I$ so that all but an $\alpha/10$-fraction of samples \new{$x \in S$}
have $v\cdot x\in I$, we know almost certainly that $v \cdot \mu \in I$,
since $v\cdot \mu$ should have at least an $\alpha/4$-fraction
of samples (coming from the clean samples) on either side of it.
This implies that samples \new{$x$ whose $v$-projections are at distance}
much further than $\sqrt{\log(1/\alpha)}$ from the endpoints of $I$ are almost certainly outliers.
Using techniques similar to the ones we discussed in Section~\ref{ssec:filter},
if the variance in the $v$-direction is more than a sufficiently large \new{constant} multiple
of $|I|^2+\log(1/\alpha)$, one can find a {\em single} subset $S'$ of $S$
so that with high probability almost all points in $S \setminus S'$ are outliers.


It remains to handle the complementary case.
If for some \new{$x \in S$} we let
$S_1=\{y\in S: v\cdot y \geq x-10\sqrt{\log(1/\alpha)}\}$ and
$S_2=\{y\in S: v\cdot y \leq x+10\sqrt{\log(1/\alpha)}\}$,
it is not hard to see that at least one of $S_1$ or $S_2$ keeps
almost all of the clean samples in $S$. In particular, if
$v\cdot \mu \geq x$, then $S_1$ will only throw out clean samples
that are at least $10\sqrt{\log(1/\alpha)}$ to the left of
$\mu$ (i.e., at most an $\alpha^{10}$-fraction).
Similarly, if $v\cdot \mu\leq x$, then $S_2$ will throw away at most
an $\alpha^{10}$-fraction of clean samples. If additionally,
(i) each of $S \setminus S_1$ and $S \setminus S_2$ contain
at least an $\alpha^2$-fraction of the total samples,
and (ii) $|S_1|^2+|S_2|^2 \leq |S|^2$, then these subsets will suffice.

The key observation is that if the variance of $S$ in the $v$-direction
is more than a sufficiently large multiple of $\log(1/\alpha)$,
we can always apply at least one of the two multi-filters described above.
This holds because if we try to apply the second multi-filter
for some given value of $x$, we will find that either the fraction
of samples with $v\cdot y\leq x-10\sqrt{\log(1/\alpha)}$
is much smaller that the fraction with $v\cdot y\leq x$
or the fraction of samples with $v\cdot y\geq x+10\sqrt{\log(1/\alpha)}$
is much smaller that the fraction with $v\cdot y\geq x$. In either case,
the tails of $v\cdot S$ must decay fairly rapidly, at least until they
are smaller than $\alpha^2$. Thus, if we cannot apply this filter for any $x$,
the contribution to the variance coming from everything except the tails
must be small. On the other hand, letting $I$ be the interval excluding the
$\alpha^2$-tails on either side, we can apply the first multi-filter
if the contribution from the \new{$\alpha^2$-tails} is large. In summary,
we can always apply one of the two multi-filters unless the variance of $v\cdot S$ is small.
Overall, this algorithm outputs $\poly(1/\alpha)$ many hypotheses,
at least one of which is within $O(\sqrt{\log(1/\alpha)/\alpha})$ of the true mean.

\subsection{Robust Estimation Using High-Degree Moments} \label{ssec:higher-moments}

The algorithms presented so far robustly estimate
the mean of high-dimensional distributions by leveraging structural information
about their covariance matrix. The robust covariance estimation
algorithm of Section~\ref{ssec:cov} uses structural information about the fourth moment tensors, but
also fits in this framework, as it works by robustly estimating the mean of the random variable $XX^T$.
It is natural to ask whether (and to what extent) one can exploit structural information about {\em higher degree moments}
of the uncorrupted distribution to robustly estimate its parameters.

For the basic case that the uncorrupted distribution is a Gaussian, algorithmically
exploiting higher-degree moment information to obtain robust estimators with information-theoretically near-optimal accuracy
turns out to be manageable. For a concrete example, we focus here on the problem of robust mean estimation
for $\mathcal{N}(\mu, I)$.
Recall that the convex programming and filtering algorithms of Section~\ref{sec:robust-mean} achieve
$\ell_2$-error $O(\eps\sqrt{\log(1/\eps)})$ in the strong contamination model,
which is optimal for spherical sub-gaussian distributions
but suboptimal (up to the $O(\sqrt{\log(1/\eps)})$ multiplicative factor) for spherical Gaussians.
The reason for this discrepancy is that the stability condition of Definition~\ref{def:stability}
and the associated filter/convex programming algorithms
only rely on the first two moments of the distribution.

The work of~\cite{DKS17-sq} showed how to leverage higher-order moment information to
improve on the $O(\eps\sqrt{\log(1/\eps)})$ error bound.
Specifically, this work gave a generalized filtering algorithm that performs ``outlier removal''
based on higher-order tensor information
(of degree $d = \Omega(\log^{1/2}(1/\eps))$)
to robustly estimate the mean of $\mathcal{N}(\mu, I)$ in the strong contamination model
within $\ell_2$-error $O(\eps)$ in time $O_\eps(d^{O(\log^{1/2}(1/\eps))})$.
(This runtime upper bound is qualitatively matched by an SQ lower bound
shown in the same paper; see Section~\ref{ssec:hardness}).
This generalized filtering and its correctness analysis leverage
properties of the Gaussian distribution, including
the a priori knowledge of the higher moments and the concentration of high-degree polynomials.

Subsequently, the work \cite{DiakonikolasKS18-nasty} (see also~\cite{DK19} for the case of
discrete distributions)
used higher-order moments to robustly learn bounded
degree polynomial thresholds functions (PTFs) under various distributions.
It should be noted that the latter result does not require knowledge of all the higher degree moments.
Specifically, the algorithm of \cite{DiakonikolasKS18-nasty} requires appropriate concentration
and anti-concentration properties, and (approximate) knowledge of the moments
up to degree $2d$, where $d$ is the degree of the
underlying PTF.

The more general setting where we only have {\em upper bounds} on the higher degree moments of the
uncorrupted distribution turns out to be substantially more challenging algorithmically.
In general, upper bounds on the higher moments
imply better information-theoretic error upper bounds for robust estimation. For example,
for robust mean estimation of distributions with bounded $k$-th central moments,
the information-theoretically optimal error is easily seen to be $\Theta(\eps^{1-1/k})$. However, it is unclear
if this error bound is attainable algorithmically for $k \geq 4$.
Without any additional assumptions on the underlying distribution (beyond
the bounded higher moments condition), recent work~\cite{HL19} gave evidence that
obtaining error $o(\eps^{1/2})$ may be computationally intractable (see Section~\ref{ssec:hardness}).

However, there are circumstances in which
higher moment information can be usefully exploited.
A number of concurrent works obtained efficient algorithms
leveraging higher-degree moments to obtain
near-optimal error guarantees~\cite{DiakonikolasKS18-mixtures, HopkinsL18, KS17, KStein17, KothariSS18}.
Specifically, the work of~\cite{DiakonikolasKS18-mixtures} gave a higher-moment
generalization of the multi-filter technique described in Section~\ref{ssec:list-dec} that leads
to a near-optimal error algorithm for list-decodable mean estimation of $\mathcal{N}(\mu, I)$.
As an application,~\cite{DiakonikolasKS18-mixtures} obtained an efficient
algorithm to learn the parameters of mixtures of spherical Gaussians under near-optimal separation between
the components. The works~\cite{HopkinsL18, KS17, KStein17, KothariSS18} used the Sums-of-Squares
meta-algorithm to obtain a number of algorithmic results, including robust mean estimation
given a Sums-of-Squares proof certifying bounded central moments~\cite{HopkinsL18, KS17},
learning mixtures of spherical Gaussians~\cite{HopkinsL18, KStein17}, and list-decodable
mean estimation~\cite{KStein17} (under a similar Sums-of-Squares certifiability assumption).
More recently,~\cite{KKK19-list, RY19-list} used the Sums-of-Squares method
to obtain the first non-trivial algorithms for list-decodable linear regression.



In this section, we describe in more detail two important settings where a fairly sophisticated use of higher degree
moments is required: list-decodable mean estimation (with near-optimal error guarantees) for Gaussian distributions;
and robust mean estimation with certifiably bounded central moments.
Our presentation will mainly focus on the methodology developed by the authors in~\cite{DiakonikolasKS18-mixtures}.
We provide a high-level overview of the Sums-of-Square approach to these problems and refer the interested reader
to the recent survey~\cite{RSS18} for a more technical exposition of this approach.

\paragraph{Near-Optimal List-Decodable Gaussian Mean Estimation.}
We consider the problem of list-decodable mean estimation, assuming
the uncorrupted samples are drawn from a spherical Gaussian distribution $\mathcal{N}(\mu, I)$.
The techniques we discussed in Section~\ref{ssec:list-dec} show how to compute
a list of $O(1/\alpha)$ hypotheses (candidate mean vectors) such that (with high probability over the uncorrupted samples)
at least one hypothesis is within $\ell_2$-norm $\tilde O(1/\alpha^{1/2})$
from the true mean $\mu$. We note that this error bound is actually very
far from the information-theoretic optimum. In particular, for a point to be a reasonable hypothesis,
there must be a cluster consisting of at least an $\alpha$-fraction
of the samples that are roughly Gaussian distributed around it.
If two such hypotheses are separated by more than a large multiple
of $\sqrt{\log(1/\alpha)}$, these clusters cannot overlap on more
than an $\alpha$-fraction of their points (by Gaussian tail bounds). However,
by a simple counting argument, this implies that there cannot be more than
$\Omega(1/\alpha)$ many such hypotheses pairwise separated by
$\Omega(\sqrt{\log(1/\alpha)})$. Hence, information-theoretically,
there exists a list of $O(1/\alpha)$ many hypotheses such that (with high probability)
at least one is within distance $O(\sqrt{\log(1/\alpha)})$ of the true mean. In fact,
this upper bound is known to be tight. In~\cite{DiakonikolasKS18-mixtures}, it is shown
how to construct distributions that are consistent with many plausible true means
each separated by $\Omega(\sqrt{\log(1/\alpha)})$.

Unfortunately, while achieving better error is information-theoretically possible,
the algorithm discussed in Section~\ref{ssec:list-dec} is not able to achieve it.
There are inherent structural reasons for this:
This algorithm attempts to find subsets of the samples with small variance.
Unfortunately, small variance is not sufficient to imply better than $O(1/\alpha^{1/2})$ error.
This is because if our $\alpha$-fraction of good samples is located at distance of $\alpha^{-1/2}$
from the other samples in some direction, the variance in that direction would only be approximately $1$,
and the true mean would differ from the sample mean by about $\alpha^{1/2}$, in $\ell_2$-norm.
Another way of putting it is that this is a question about concentration.
Our existing algorithm manages to produce a set of samples, at least an $\alpha$-fraction of which
are good, which also has bounded variance.
Since bounded variance ensures some measure of concentration, this ensures that
the mean of the clean samples in this set (close to the true mean) cannot be too far
from the sample mean of this set. Unfortunately, the bounded variance condition
can only take us so far, and in fact is consistent with the mean of the good samples
being as far as $\alpha^{1/2}$ from the sample mean.

To improve our error guarantee, we want to find a set of samples with stronger concentration bounds.
A natural way to do this is to ensure that our set of samples has bounded higher moments.
Specifically, we say that a set of points $S$ has bounded \emph{$d^{th}$ central moments}
if for every unit vector $v$ we have that $\E_{x\sim_u S}[|v\cdot(x-\mu)|^d] \leq C$, for some constant $C>0$.
We note that if $S$ has bounded $d^{th}$ central moments, the mean
of any $\alpha$-fraction of the points is no more than $(C/\alpha)^{1/d}$ far from the overall mean.
Thus, if we could find a set with a large fraction of clean points which also had bounded
$d^{th}$ central moments for some $d\geq 2$, we would be able to obtain better error bounds.

The above paragraph naturally gives the outline for a better algorithm.
We start with some set $S$ of samples.
If its $d^{th}$ central moments are small, we return the mean of $S$.
Otherwise, we find some direction $v$ which has large central moments,
project onto that direction and use this information to create a multi-filter and repeat.

Unfortunately, it is not clear how to implement this algorithm efficiently.
Recent results show that determining whether a {\em generic} point set 
has bounded central moments, for $d>2$, seems to be computationally intractable~\cite{HL19}.
However, we are not dealing with arbitrary point sets. We can take advantage of the fact
that Gaussian distributions will satisfy stronger conditions on their higher moments
than simply bounded central moments, and we can design algorithms that attempt
to find sets of samples satisfying these more stringent (and hopefully computationally checkable)
conditions instead. There are two different ways to implement this idea:
the squares of polynomials method~\cite{DiakonikolasKS18-mixtures}
and the Sums-of-Squares method.

For the squares of polynomials method, we note that the $d=2$ case is easy,
as it amounts to finding the maximum value of a quadratic polynomial on a sphere
(which can be solved by spectral methods). For $d>2$, the problem requires that
we optimize a higher-degree polynomial, which is not that simple. However, if our
good samples are Gaussian with a given mean, we know what the expectations
of higher degree polynomials should be and can thus check for it. In particular,
in order to verify that the sample set has bounded $2d^{th}$ central moments,
one can check if there is any degree-$d$ polynomial $p$ where the average
value of $p^2$ is too large. This can be checked by considering
the problem of optimizing a quadratic form over polynomials $p$,
and it is sufficient by considering $p(x)=(v\cdot(x-\mu))^d$. If such a $p$ is found,
we can use it to construct a multi-filter, though the mechanism for doing
so is highly non-trivial and invokes many special properties of the Gaussian distribution.
The interested reader can find the full details in~\cite{DiakonikolasKS18-mixtures}.

The other method makes use of the Sums-of-Squares proof system.
In particular, it is shown that the Gaussian bounded central moments
can be proved in the Sum-of-Squares proof system. Thus, the goal is
to find subsets of the samples for which a similar Sums-of-Squares
proof allows one to show bounded central moments.
These techniques were developed in~\cite{HopkinsL18,KothariSS18}.

For list-decodable mean estimation, both of these methods allow
one to learn the mean of a Gaussian to $\ell_2$-error approximately $\alpha^{-1/(2d)}$
by considering the $2d^{th}$ moments. These algorithms will have runtime $\mathrm{poly}(n^d)$
and, by making $d$ super-constant, we can in fact obtain poly-logarithmic
error in quasi-polynomial time. The Sums-of-Squares
method has the advantage that it is much more general and applies
not just to Gaussians but to any distribution whose bounded central moments
can be certified by Sums-of-Squares proofs of small degree.
However, these systems must search for Sums-of-Squares proofs,
which require solving large convex optimization problems,
meaning that these algorithms will be slower in practice.

\paragraph{Robust Mean Estimation with Certifiably-Bounded Higher Moments.}
Another interesting application of the Sums-of-Squares method in this context
is in reducing the conditions needed for robust mean estimation. Recall that the
definition of stability (Definition~\ref{def:stability}) has two conditions: First, it requires
strong concentration on the uncorrupted samples, to ensure that
removing any small fraction does not substantially alter the mean
(a condition that is information-theoretically necessary in some contexts).
Second, to satisfy the definition for $\delta = o(\sqrt{\eps})$, the algorithm needs to know
the covariance matrix of the good samples. This latter condition
is not required for information-theoretic reasons, but for computational ones.
The algorithm needs to know the covariance matrix of the clean samples
so that it can detect even small increases in this covariance caused by corrupted samples.

The above implies for example, that if we have a distribution with, say, bounded fourth
central moments, the first stability condition holds
with $\delta = O(\eps^{3/4})$; while the second part will only hold
with $\delta=\Omega(\sqrt{\eps})$, as the actual covariance matrix with no errors
might be modestly far from the identity. The first condition implies
that information-theoretically, it should be possible to learn the mean
to error $O(\eps^{3/4})$ (for example by looking at a truncated mean in each direction),
but the standard filter will not get error better than $\Omega(\sqrt{\eps})$,
as an adversary can add errors to keep the covariance matrix at most the identity
and yet corrupt the sample mean by this much.

To circumvent this natural bottleneck, one would similarly need a way to take advantage of higher moments
and leverage this stronger concentration. Once again, the Sums-of-Squares method can be used to achieve this
under certain assumptions.
Roughly speaking, if the uncorrupted distribution has bounded $d^{th}$ central moments
\emph{provable by low degree Sums-of-Squares proofs}, then by searching
for sets of samples with such a proof, one can obtain error $O_d(\eps^{1-1/d})$.

\subsection{Fast Algorithms for High-Dimensional Robust Estimation} \label{ssec:fast}


The main focus of the recent algorithmic work in this field has been on obtaining
polynomial-time algorithms for various high-dimensional robust  estimation tasks.
Once a polynomial-time algorithm for a computational problem is discovered,
the natural next step is to develop faster algorithms for the problem -- with linear time as the ultimate goal.
While recent work has led to polynomial-time robust learners
for several fundamental tasks, these algorithms are significantly slower than their non-robust
counterparts (e.g., the sample average for mean estimation).
This raises the following question:
{\em Can we design robust estimators that are as
efficient as their non-robust analogues?}
In addition to its potential practical implications, making progress on this general direction
is of fundamental theoretical interest, as it can elucidate the effect of the robustness
requirement on the computational complexity of high-dimensional statistical learning.

The above direction was initiated by~\cite{ChengDR18} in the context of robust mean estimation.
More specifically,~\cite{ChengDR18} gave a robust mean estimation algorithm for bounded
covariance distributions on $\R^d$ that is near-minimax optimal,
achieves the optimal $\ell_2$-error guarantee of $O(\sqrt{\eps})$, and
runs in time $\tilde{O}(n d)/\poly(\eps)$, where $n$ is the number of samples.
That is, the algorithm of~\cite{ChengDR18} has the same (optimal) sample complexity and error
guarantee as previous polynomial-time algorithms~\cite{DKKLMS16, DKK+17}, while
running in near-linear time when the fraction of outliers $\eps$ is a small constant.
At the technical level,~\cite{ChengDR18} builds on the convex programming
approach of Section~\ref{ssec:convex-program}.

Subsequent work~\cite{DepLec19} observed that a simple preprocessing step allows one
to reduce to the case when the fraction of corruptions is a small universal constant.
As a corollary, it was shown in~\cite{DepLec19} that
a simple modification of the~\cite{ChengDR18} algorithm runs in $\tilde{O}(n d)$ time.
More importantly,~\cite{DepLec19} gave a probabilistic analysis that leads to
a fast mean estimation algorithm that is simultaneously outlier-robust and
achieves sub-gaussian tail bounds. (We note that the question of designing estimators
for the mean of heavy-tail distributions with sub-gaussian rates has been of substantial interest
recently. The interested reader is referred to~\cite{Hop18, CherapanamjeriF19, LugosiM19, LugosiM19robust})
for recent developments on this topic.) Independently and concurrently to~\cite{DepLec19},
\cite{DHL19} built on the filtering framework of Section~\ref{ssec:filter} to give
a different $\tilde{O}(n d)$ time robust mean estimation algorithm.
Moreover,~\cite{DHL19} gave an empirical evaluation
showing the practical performance of their algorithm.

Prior to the aforementioned developments, the fastest known runtime for robust mean estimation
was $\tilde{O}(n d^2)$ and was achieved by the filtering algorithm of Section~\ref{ssec:filter}.
While we did not provide a detailed runtime analysis in Section~\ref{ssec:filter}, it is not hard to show
that each filter iteration can be implemented in $\tilde{O}(n d)$ time (using power iteration)
and that the number of iterations can be bounded from above by $O(d)$.
While the filtering algorithm has been observed to run very fast in practice, taking at most a small constant number
of iterations on real datasets~\cite{DKK+17, DiakonikolasKKLSS2018sever},
one can construct examples where an $\eps$-fraction of outliers force the algorithm
to take $\Omega(d)$ iterations. This can be achieved by placing the outliers in $\Omega(d)$ orthogonal
directions at appropriate distances from the true mean.
Conceptually, the bottleneck of the filtering method is that it relies on a certificate (Lemma~\ref{lem:stability})
that allows the algorithm to remove outliers from one direction at each iteration.

A detailed explanation and comparison of the techniques in~\cite{ChengDR18, DepLec19, DHL19}
is beyond the scope of this survey. At a high-level, a conceptual commonality of these works is that they
leverage techniques from continuous optimization to develop iterative methods
(with each iteration taking near-linear time) that are able to deal with multiple directions
{\em in parallel}. In particular, the total number of iterations in each of these methods is
at most {\em poly-logarithmic} in $d/\eps$.


Beyond robust mean estimation, the work~\cite{CDGW19} recently studied
the problem of robust covariance estimation with a focus
on designing faster algorithms. By building on the techniques of~\cite{ChengDR18},
they obtained an algorithm for this problem with runtime $\tilde{O}(d^{3.26})$.
Rather curiously, this runtime is not linear in the input size, but nearly matches
the (best known) runtime of the corresponding non-robust estimator (i.e., computing the empirical covariance).
Intriguingly, \cite{CDGW19} also provided evidence that the runtime of their algorithm
may be best possible with current algorithmic techniques.


\subsection{Computational-Statistical Tradeoffs in Robust Estimation} \label{ssec:hardness}

The golden standard in robust estimation is to design computationally efficient
algorithms with optimal sample complexities and error guarantees.
A conceptual message of the recent body of algorithmic work in this area
is that robustness may not pose computational impediments to high-dimensional estimation.
Indeed, for a range of fundamental statistical tasks, recent work has developed computationally efficient
robust estimators with {\em dimension-independent} (and, in some cases, near-optimal) error guarantees.
However, it turns out that, in some settings, robustness creates {\em computational-statistical tradeoffs}.
Specifically, for several natural high-dimensional robust estimation tasks, we now have compelling evidence that
achieving, or even approximating, the information-theoretically optimal error is computationally intractable.

Progress in this direction was first made in~\cite{DKS17-sq}, which used the framework of
Statistical Query (SQ) algorithms~\cite{kea93} to establish computational-statistical trade-offs
for a range of robust estimation tasks involving Gaussian distributions.
More specifically, it was shown in~\cite{DKS17-sq}
that even for the basic problems of robust mean and covariance estimation of a high-dimensional
Gaussian with contamination in total variation distance,
achieving the optimal error requires super-polynomial time.
The same work established computational-statistical trade-offs for the problem of robust sparse mean estimation,
even in Huber's contamination model,
showing that efficient algorithms require quadratically more samples than the information-theoretic minimum.
Interestingly, both these SQ lower bounds are matched by the performance of recently developed
robust learning algorithms~\cite{DKKLMS16, BDLS17}.

Motivated by this progress,~\cite{HL19} made a first step towards proving computational lower bounds for robust mean estimation based on {\em worst-case hardness} assumptions. In particular, this work established that current algorithmic techniques for robust mean estimation may not be improvable in terms of their error guarantees, in the sense that they stumble upon a well-known computational barrier -- the so-called small set expansion hypothesis (SSE), closely related to the unique games conjecture (UGC). More recently,
~\cite{BB19} proposed a $k$-partite variant of the planted clique problem~\cite{Jerrum92} and gave a reduction,
inspired by the SQ lower bound of~\cite{DKS17-sq}, implying (subject only to the hardness of the proposed problem) a statistical--computational gap for this problem. An interesting open problem is to establish compelling evidence (e.g., in the form
of a Sums-of-Squares lower bound) of the {\em average-case hardness} of the proposed $k$-partite variant of planted clique.

\medskip

\noindent In the following paragraphs, we provide a more detailed description of~\cite{DKS17-sq, HL19}.

\paragraph{Statistical Query Lower Bounds.}
For the statistical estimation tasks studied in this article,
the input is a set of samples drawn from a probability distribution of interest.
Statistical Query (SQ) algorithms~\cite{kea93} are a restricted class of algorithms
that  are only allowed to (adaptively) query expectations of bounded
functions of the distribution -- rather than directly access samples.

In particular, if $f$ is any bounded function, one may attempt to approximate
$\E[f(X)]$ by taking samples from the distribution $X$. With $O(1/\tau^2)$ samples,
one can obtain the correct answer within additive accuracy $\tau$ with high probability.
By doing this for several different functions $f$, perhaps chosen adaptively,
one can try to learn properties of the underlying distribution $X$. In the SQ model, the algorithm
gets to ask queries (with some given accuracy $\tau$) of an oracle.
These queries are in the form of a function $f$ with range contained in $[-1,1]$ and a desired accuracy $\tau$.
The oracle then returns $\E[f(X)]$ to accuracy $\tau$, and the algorithm
gets to (adaptively) chose another $f$ up to $Q$ times. Roughly speaking, an SQ
algorithm with accuracy $\tau$ and $Q$ queries corresponds
to an actual algorithm using $O(1/\tau^2)$ samples and $O(Q)$ time.

The Statistical Query (SQ) model is actually quite powerful:
a wide range of known algorithmic techniques in machine learning are known to be implementable using SQs.
These include spectral techniques, moment and tensor methods, local search (e.g., Expectation Maximization),
and many others (see, e.g.,~\cite{Feldman13}). In fact, nearly every known statistical algorithm
(with a small number of exceptions) with provable performance guarantees can be simulated
with small loss of efficiency in the SQ model. This makes the SQ model very useful for proving lower bounds,
as a lower bound in this model applies to a broad family of algorithms.


It is easy to see that one can estimate moments and approximate medians in the SQ model.
In fact, without difficulty, various filter algorithms described in this article
can be implemented in the SQ framework. Indeed, moment computations
allow one to estimate the covariance matrix. If large eigenvalues are found,
further measurements can approximate the cumulative density distribution
of the projection and decide on a filter. From then on, measurements can
be made conditional on passing the filter (by measuring $f(x)$ times the indicator
function of $x$ passing the filter).


A recent line of work~\cite{Feldman13, FeldmanPV15, FeldmanGV15, Feldman16}
developed a framework of SQ algorithms for search problems over distributions.
One can prove unconditional lower bounds on the computational complexity of
SQ algorithms via a notion of {\em Statistical Query dimension}.
This complexity measure was introduced in~\cite{BFJ+:94} for PAC learning of Boolean functions and was recently generalized to the unsupervised setting~\cite{Feldman13, Feldman16}. A lower bound on the SQ dimension of a learning problem provides an unconditional lower bound on the computational complexity of any SQ algorithm for the problem.
Suppose we want to estimate the parameters of an unknown distribution $X$ that belongs in a known family
$\mathcal{D}$. Roughly speaking, the aforementioned work has shown that if there are many possible
distributions $X \in \mathcal{D}$ whose density functions are pairwise nearly orthogonal with respect to an appropriate inner product,
any SQ algorithm with insufficient accuracy will require many queries to determine
which of these distributions it is sampling from.

The work of~\cite{DKS17-sq} gives an SQ lower bound for the statistical task of
{\em non-Gaussian component analysis (NGCA)}~\cite{BlanchardKSSM06}. Intuitively, this is
the problem of finding a non-Gaussian direction in a high-dimensional dataset.
In more detail, let $A$ be the pdf of a univariate distribution with the property
that the first $m$ moments of $A$ match the corresponding moments of the standard univariate
Gaussian $\mathcal{N}(0, 1)$. For a unit vector $v\in \R^d$,
let $\mathbf{P}_v$ be the pdf of the distribution on $\R^d$ defined as follows:
The projection of $\mathbf{P}_v$ in the $v$-direction is equal to $A$, and $\mathbf{P}_v$
is a standard Gaussian in the orthogonal complement $v^{\perp}$.
Given sample access to a distribution $\mathbf{P}_{v^{\ast}}$, for some unknown direction (unit vector)
$v^{\ast} \in \R^d$, the goal of NGCA is to approximate the hidden vector $v^{\ast}$.
It is shown in~\cite{DKS17-sq} that any SQ algorithm to approximate $v^{\ast}$
requires either queries of accuracy $d^{-\Omega(m)}$
or exponentially many $2^{d^{\Omega(1)}}$ oracle queries.

The aforementioned SQ lower bound can be used in essentially a black-box manner
to obtain nearly tight SQ lower bounds for a range of high-dimensional estimation tasks
involving Gaussian distributions, including learning mixtures of Gaussians, robust
mean and covariance estimation, and robust sparse mean estimation. At a high-level,
this is achieved by constructing instances of these problems that amount
to an NGCA instance for an appropriate parameter $m$ of matching moments.
The main idea is to add noise to the distribution in question so that the noisy distribution
is of the form $\mathbf{P}_{v^{\ast}}$ for some moment-matching distribution $A$. Using these techniques,
it was shown that super-polynomial complexity (in terms of either number of queries or accuracy) is
required to learn the mean of an $\eps$-corrupted Gaussian to $\ell_2$-error $o(\eps\sqrt{\log(1/\eps)})$
or to learn its covariance to Frobenius error $o(\eps\log(1/\eps))$, both of which are tight~\cite{DKKLMS16}.
It was also be shown that to robustly learn the mean of a $k$-sparse Gaussian
to constant $\ell_2$-error requires either super-polynomially many queries or queries with accuracy $O(\max(1/k,1/\sqrt{d}))$,
which morally corresponds to taking at least $\Omega(\min\{d,k^2\}$ samples, nearly
matching the sample complexity of known algorithms~\cite{BDLS17}.

\paragraph{Reductions from Worst-case Hard Problems.}
Proving computational lower bounds via reductions from worst-case hard problems has proved
to be a rather challenging goal for statistical estimation tasks. Some recent progress was made on this front
in the context of robust mean estimation. Specifically, the work~\cite{HL19} established computational
lower bounds against algorithmic techniques operating along the same lines as existing ones.
In particular, existing algorithms for robust mean estimation depend on being able to find
a computationally verifiable certificate that (under appropriate conditions) implies
that the sample mean is close to the true mean (analogous to Lemma~\ref{lem:stability-full}).
It is shown in \cite{HL19} that, in some cases, finding natural certificates may be computationally intractable.

As a specific example, consider the class of distributions on $\R^d$ with bounded fourth central moments.
For such distributions, it is not hard to show that the trimmed mean correctly estimates
the mean of any one-dimensional projection within error $O(\eps^{3/4})$,
showing that this error rate is information-theoretically optimal, within constant factors.
However, for an algorithm to achieve this, using techniques like those already known,
it would need to have a way to certify whether or not the sample mean of a given point set
is close to the true mean. A natural way to do this would involve verifying whether the point set
itself has bounded fourth central moments. However, \cite{HL19} show that, subject to the Small-Set-Expansion Hypothesis,
this is computationally intractable. In fact, the certification problem remains intractable
even if one needs to distinguish between a distribution which has many bounded central moments
and one that lacks bounded fourth central moments.
While this hardness result is hardly definitive (as it leaves space for a variety of different kinds of algorithms),
it excludes some of the most natural approaches of extending existing techniques.



\section{Conclusions} \label{sec:conc}


In this article, we gave an overview of the recent developments on algorithmic aspects of high-dimensional robust statistics.
While substantial progress has been made in this field during the past few years, the results obtained so far merely scratch
the surface of the possibilities ahead. A major goal going forward is to develop a {\em general algorithmic theory
of robust estimation}. This involves (1) developing novel algorithmic techniques that lead to efficient robust estimators
for more general probabilistic models and estimation tasks; (2) obtaining a deeper understanding of the computational limits
of robust estimation; (3) developing mathematical connections to related areas, including non-convex optimization and privacy;
and (4) exploring applications of algorithmic robust statistics to exploratory data analysis,
safe machine learning, and deep learning.


One of the main conceptual contributions of the classical theory of robust statistics
has been to challenge traditional statistical assumptions about the data generating process,
thereby enabling the design of methods that are stable to deviations
from these assumptions. The precise form of such deviations depends
on the setting and can give rise to various models of robustness. We believe that a central objective
in a modern theory of robustness is to rethink old models and develop new ones that enable the design of
new practical algorithms with provable guarantees.


\section*{Acknowledgments}
We thank Alistair Stewart for extensive collaboration in this area.
We thank Tim Roughgarden, Ankur Moitra, and Greg Valiant for useful
comments on a shorter version of this article.

\bibliographystyle{plain}
\bibliography{allrefs}

\end{document}